\documentclass[pra,aps,twocolumn,superscriptaddress,longbibliography]{revtex4-1}
\usepackage[colorlinks=true, citecolor=red, urlcolor=blue ]{hyperref}
\usepackage{graphicx}
\usepackage{bm}
\usepackage{setspace}
\usepackage{amsmath,amsfonts}
\usepackage{amsthm}
\usepackage{hyperref}
\usepackage{mathrsfs}
\usepackage{xcolor}
\hypersetup{
    urlcolor=magenta,
    citecolor=blue
           }
\usepackage{comment}
\usepackage{braket}
\theoremstyle{plain}
\usepackage{color}
\usepackage{amssymb}
\usepackage{amsthm}
\usepackage{amsfonts}
\usepackage{float}
\usepackage{tabularx}
\usepackage{graphicx}
\usepackage[export]{adjustbox}
\usepackage{mathtools}
\usepackage{esvect}
\usepackage{wrapfig}
\usepackage{amsthm}
\usepackage{verbatim}
\usepackage{bbm}
\usepackage[normalem]{ulem}

\usepackage{enumitem}
\usepackage{fmtcount}
\usepackage{booktabs}
\usepackage{csquotes}
\usepackage{epsfig}

\usepackage{tabularx}
\usepackage{graphicx}
\usepackage{amsmath}
\usepackage{braket}
\usepackage{latexsym}
\usepackage{bm}
\usepackage{graphics,epstopdf}
\usepackage{enumitem}
\usepackage{fmtcount}
\usepackage{booktabs}
\usepackage{csquotes}
\usepackage{epsfig}

\theoremstyle{plain}

\def\bea{\begin{eqnarray}}
\def\eea{\end{eqnarray}}
\def\ba{\begin{array}}
\def\ea{\end{array}}

\def\beq{\begin{equation}}
\def\eeq{\end{equation}}

\usepackage[normalem]{ulem}
\usepackage{float}
\usepackage{graphicx}  
\usepackage{dcolumn}          
\usepackage{amssymb}
\usepackage{appendix}
\usepackage{physics}   
\usepackage{mathtools}
\usepackage{esvect}
\usepackage{wrapfig}
\usepackage{amsthm}
\usepackage{verbatim}
\usepackage{bbm}

\usepackage[mathscr]{euscript}
\def\Tr{\operatorname{Tr}}

\def\({\left(}
\def\){\right)}
\def\[{\left[}
\def\]{\right]}

\newtheorem{proposition}{Proposition}

\begin{document}



\title{Non-Markovianity vs athermality: perturbation-enhanced information backflow}

\author{Anindita Sarkar}
\affiliation{Harish-Chandra Research Institute,  A CI of Homi Bhabha National Institute, Chhatnag Road, Jhunsi, Prayagraj  211019, India}
\author{Debarupa Saha}
\affiliation{Harish-Chandra Research Institute,  A CI of Homi Bhabha National Institute, Chhatnag Road, Jhunsi, Prayagraj  211019, India}
\author{Ujjwal Sen}
\affiliation{Harish-Chandra Research Institute,  A CI of Homi Bhabha National Institute, Chhatnag Road, Jhunsi, Prayagraj  211019, India}
\begin{abstract}
Non-Markovianity and athermality are 
useful resources in quantum technologies, and it is therefore important to understand the relations between the two, for general quantum dynamics. We propose three measures of non-Markovianity, first within the ambit of thermal operations, and then beyond it, that result from 
unavoidable perturbations in system's Hamiltonian and that leads to  violations of conservation of total energy characterizing any thermal operation. The proposed measures are based respectively on  system-environment entanglement, total correlation in the system-environment partition, and on a concept of distance defined on the sets of usual  and approximate thermal operations. We investigate the response of non-Markovianity to the athermality-inducing perturbations, using all the three measures. For the entanglement and  distance-based measures, we derive upper bounds on the response by a quantity that depends on the perturbative  Hamiltonian. 
For the total correlation-based measure, we are able to compute the exact response. We present examples of qubit-qubit and qubit-qutrit systems for which perturbation leads to enhancement of non-Markovianity, as quantified by the entanglement and total correlation-based measures.
 
\end{abstract}
\maketitle
\section{\MakeUppercase{Introduction}}
 The theory of open quantum systems~\cite{lidar2020lecturenotestheoryopen,Rivas_2012,10.1093/acprof:oso/9780199213900.001.0001}, describes the evolution of a quantum system, while it interacts with an external environment. Such evolutions  can be Markovian or non-Markovian~\cite{PhysRevA.90.052118,CHANDA20161,PhysRevLett.120.060406,PhysRevLett.118.080404,PhysRevLett.118.120501,PhysRevLett.118.050401,PhysRevLett.119.190401,PhysRevLett.117.050403,PhysRevLett.99.160502,PhysRevLett.104.250401,doi:10.1142/S1230161218500142,PhysRevA.95.012122,PhysRevA.96.052125,Vacchini_2011,RevModPhys.88.021002,RevModPhys.89.015001}. Markovian evolutions are quintessential to scenarios in which the interaction between the system and the environment is very weak. As a consequence, the system and the environment remains uncorrelated throughout the evolution. There is only a unidirectional flow of information from the system to the environment. The environment possesses no memory about the initial state of the system and therefore, no backflow of information is plausible in a Markovian evolution. On the other hand, non-Markovian evolutions are the result of strong interactions between the system and the environment, due to which they may become correlated during the evolution. The environment retains memory about the states of the system, and the flow of information is bidirectional. In fact, information backflow from environment to system is a key feature of non-Markovian evolutions.
 
Non-Markovianity has proven to be a vital resource in various fields of quantum information and computation. Non-Markovianity can help in assisting steady-state entanglement formation~\cite{PhysRevLett.108.160402}. Entangled states can metrologically outperform uncorrelated states with same resource using non-Markovian dynamics~\cite{PhysRevLett.109.233601} and they can surpass the standard quantum limit under non-Markovian noise~\cite{PhysRevA.84.012103}. Harnessing non-Markovianity can enhance capacities of quantum channels~\cite{bylicka2013nonMarkovianityresourcequantumtechnologies}. Perfect teleportation can be made possible with the help of non-Markovianity even with mixed initial states formed due to decoherence~\cite{Laine2014}. In addition to these, non-Markovianity can be fundamentally useful in increasing the security of continuous-variable key distribution protocols~\cite{PhysRevA.83.042321}, improving the accuracy in simulating quantum dynamics~\cite{chen2023nonmarkovianitybenefitsquantumdynamics}, enhancing thermodynamic work extraction~\cite{Bylicka2016}, providing communication advantages in quantum switches~\cite{PhysRevA.106.052410} and improving the efficiency in teleporting quantum correlated states~\cite{Motavallibashi:21}. 
This versatility of non-Markovian quantum evolutions motivates the  quantification of non-Markovianity of generic open system evolutions, in as many ways as possible, and also to relate it with other physical quantities.

In the current literature, a number of measures exist for the quantification of non-Markovianity. Examples include those 
based on trace distance~\cite{PhysRevLett.103.210401,PhysRevA.81.062115,PhysRevA.92.042108,RevModPhys.88.021002}, fidelity~\cite{PhysRevA.82.042107,PhysRevA.84.052118}, divisibility of the dynamical map~\cite{PhysRevLett.105.050403,PhysRevA.83.052128}, quantum mutual information~\cite{PhysRevA.86.044101} and quantum Fisher information~\cite{PhysRevA.82.042103}. For other measures, see references~\cite{PhysRevLett.101.150402,PhysRevA.83.062115,Breuer_2012,PhysRevA.87.042103,PhysRevA.88.020102,PhysRevLett.112.120404,PhysRevA.89.042120,PhysRevA.89.052119,PhysRevLett.112.210402,Bylicka2014,Rivas_2014,PhysRevA.90.042101,PhysRevA.91.032115,PhysRevLett.116.020503,PhysRevA.93.022117,He_2017,PhysRevA.96.022106,LI20181,PhysRevA.99.012303,anand2019quantifyingnonMarkovianityquantumresourcetheoretic,PhysRevA.100.012120,UTAGI2021126983,10.3389/frqst.2023.1134583}. 

Taking into account the benefits of non-Markovianity as a resource, we explore non-Markovianity in the context of a special type of operations called
thermal operations ($\mathbb{TO}$)~\cite{janzing2000thermodynamiccostreliabilitylow,PhysRevLett.111.250404,Horodecki_2013,rf} and approximate thermal operations ($\mathbb{TO_{\epsilon}})~$\cite{one}. These sets of operations form  important and physically relevant subclasses within generic open quantum evolutions. $\mathbb{TO}$ forms the class of free operations in the resource theory of thermodynamics~\cite{PhysRevLett.111.250404,doi:10.1142/S0217979213450197,Ng2018,Lostaglio_2019}. $\mathbb{TO}$ on a quantum system is generated by the action of energy-preserving unitaries on the combined system-environment state, followed by tracing out the environment, whereas $\mathbb{TO}_{\epsilon}$ is generated by leveraging the energy-preserving feature of the unitaries corresponding to $\mathbb{TO}$. Such an evolution may result from small perturbations in the system Hamiltonian. The degree of perturbation is denoted by $\epsilon$.

We have quantified non-Markovianity for both the types of operations, 
using three measures. The first measure is based on the amount of entanglement created between the system and the environment  for both $\mathbb{TO}$ and $\mathbb{TO}_{\epsilon}$. 
As the second  measure, we consider the amount of  quantum mutual information generated between the system and environment during the evolution via $\mathbb{TO}$ or $\mathbb{TO}_{\epsilon}$. Our third measure is distance-based, and is given by the trace norm of the difference of a state obtained by applying any map in the set of $\mathbb{TO}$/$\mathbb{TO}_{\epsilon}$, with its closest state obtained by applying Markovian thermal operations/approximate Markovian thermal operations ($\mathbb{MTO}$/$\mathbb{MTO}_\epsilon$), maximized over all possible input states. 

It is particularly interesting to examine the response of the non-Markovianity to the introduction of perturbation in the system. We determine this analytically by computing the difference, for all the non-Markovianity measures, corresponding to $\mathbb{TO}$$_\epsilon$ with their $\mathbb{TO}$ counterparts. Using perturbation theory, we are able to derive upper bounds on the differences, considering entanglement and distance-based measures. In case of the correlation-based measure, we are able to exactly compute the difference. Next, we present  examples, considering qubit-qubit and qubit-qutrit scenarios, demonstrating the existence of cases where an enhancement in the non-Markovianity due to perturbations is observed. These examples indicate that initial perturbations in system Hamiltonian can be beneficial in terms of resource (non-Markovianity) generation. This is akin to 
the order from disorder instances in 
cooperative phenomena~\cite{PhysRevA.84.042334,Sadhukhan_2015,PhysRevE.93.012131,PhysRevE.93.032115,PhysRevB.94.014421,Mishra_2016,Bera_2019,PhysRevA.101.053805,Sarkar_2022,PhysRevA.110.012620}. In our case, we show that a deviation from the ideal condition of $\mathbb{TO}$, typically termed as \emph{athermality}, may enhance non-Markovianity.


These are the key ideas on which our work is based. The rest of the paper is organized in the following way. Sec. \ref{preliminaries} is devoted to brief discussions about Markovian processes, and the concepts of $\mathbb{TO}$, $\mathbb{MTO}$, $\mathbb{TO}_{\epsilon}$, $\mathbb{MTO}_\epsilon$. Sec.~\ref{Conds} discusses the conditions for state transformations under $\mathbb{MTO}$ and $\mathbb{MTO}_\epsilon$. Sec.~\ref{nm} provides a succinct account on some existing measures of non-Markovianity, as well as our proposed measures. In Sec.~\ref{measure}, relations connecting the measures corresponding to the unperturbed and  perturbed cases are derived, and supplemented with examples. We conclude in Sec.~\ref{conc}.
\section{Preliminaries}\label{preliminaries}
\subsection{Thermal operations}
Here we introduce the concept of thermal operations~\cite{janzing2000thermodynamiccostreliabilitylow} ($\mathbb{TO}$). Suppose we have a system $S$ in the state $\rho_S$ and environment $B$ in the thermal state $\tau_B=\frac{e^{-\beta H_B}}{\Tr(e^{-\beta H_B})}$, where $\beta=\dfrac{1}{T}$ is the inverse temperature of the environment. Initially, the state of the system is uncorrelated with that of the environment. Then, a global unitary $U_{SB}$ acts on the system and the environment to couple them together. The state of the system after the application of the global unitary is obtained by tracing out the environment after the evolution, 
\begin{equation*}
\Lambda(\rho_S)=\Tr_B(U_{SB}\rho_S \otimes \tau_B U_{SB}^{\dagger}).  
\end{equation*}
The operation $\Lambda$ will be called thermal operation if the global unitaries generating the operation are energy-preserving, 
    \begin{equation*}
[U_{SB},H_T]=0,   
\end{equation*}
      where $H_T=H_S\otimes \mathcal{I}_B+ \mathcal{I}_S\otimes H_B$ is the total system-environment Hamiltonian. $H_S$ and $H_B$ are the system and environment Hamiltonian respectively. And $\mathcal{I}_S$ and $\mathcal{I}_B$ denote identity operators acting in the Hilbert spaces of the system and environment respectively. 
      
Also, the operation maps a thermal state of the system (with same temperature as that of the environment) to the same thermal state,  
\begin{equation*}
\Lambda(\tau_S)=\tau_S.
\end{equation*}
Where $\tau_S=\frac{e^{-\beta H_S}}{\Tr(e^{-\beta H_S})}$ denotes a thermal state of the system, having the same temperature as that of the environment.

 In the subsequent subsection, we give a brief account of the Markovian processes and discuss the key assumptions.
 \vspace{-0.9cm}
 \subsection{Markovian processes}
 Let the system under consideration be in an initial  state $\rho_S$, defined in a Hilbert space $\mathbb{H}_S$ of dimension $d_1$. The system interacts with a environment in an initial state $\rho_B$,  defined in a Hilbert space $\mathbb{H}_B$ of dimension $d_2$ . The dynamics describing the evolution of the system's state is considered to be Markovian~\cite{10.1093/acprof:oso/9780199213900.001.0001}, if it has the following characteristics. Firstly, the
 flow of information during the evolution is unidirectional, that is, information can flow from the system to the environment, but any backflow from the environment to the system does not occur. Secondly, the interaction between the system and the environment must be weak enough so that the system and the environment remains practically uncorrelated through out the evolution. This implies that the combined state of the system-environment at any particular instant $t$ is given as,
 \begin{equation*}
     \rho_{SB}(t) \approx \rho_S(t) \otimes \rho_B. 
 \end{equation*}
 Finally, the environment is taken to be much larger as compared to the system, such that during the evolution the environment's state changes negligibly. At any instant $t$, environment's state is given as,
  \begin{equation*}
     \rho_B(t) \approx \rho_B.
 \end{equation*}
 We introduce the concept of Markovian thermal operations in the next subsection. They are obtained by imposing the constraint of Markovianity on the elements of the set of all thermal operations.
 \subsection{Markovian thermal operations}
We define the set of Markovian thermal operations ($\mathbb{MTO}$) as those thermal operations which keep the total state of the system and the environment state in a tensor product form throughout the evolution occurring due to the application of the global unitary, $U_{SB}$. Also, we consider the state of the environment to remain unchanged during the evolution under $\mathbb{MTO}$, i.e.
\begin{equation*}
U_{SB}\rho_S \otimes \tau_B U_{SB}^{\dagger}=\rho'_S \otimes \tau_B.  
\end{equation*}
Where,
\begin{equation}
   \rho'_S=\Tr_B\left(U_{SB}\rho_S \otimes \tau_B U_{SB}^{\dagger}\right). \label{rho}
\end{equation}
This is a scenario similar to the non-Markovian collisional model~\cite{PhysRev.129.1880,CICCARELLO20221} where at every time of observation, the environment state so obtained after the evolution is discarded and replaced by freshly prepared state of the environment equivalent to the initial state $\tau_B$. This ensures that all the conditions of Markovianity are satisfied.

Clearly, the set $\mathbb{MTO}$ is a subset of the set of $\mathbb{TO}$. This is pictorially depicted in fig.~\ref{f1}. 

Due to perturbation in the system, the sets of $\mathbb{TO}$ and $\mathbb{MTO}$ are extended to form the set of approximate thermal operations $\mathbb{TO}_{\epsilon}$ and approximate Markovian thermal operations $\mathbb{MTO}_\epsilon$ respectively.
We elaborate about the set $\mathbb{TO}_{\epsilon}$ in the following subsection.
 \subsection{Approximate thermal operations}
 The system's Hamiltonian, in presence of perturbations, assumes the following form:
  \begin{equation*}
H'_S=H_S+\epsilon H',
\end{equation*}
So the total Hamiltonian changes to 
 \begin{align*}
H'_T&=H'_S\otimes \mathcal{I}_B+ \mathcal{I}_S\otimes H_B,\nonumber\\
&=H_T+\epsilon H'\otimes\mathcal{I}_B.
\end{align*}
Then the commutator between the global unitary $U_{SB}$ and the new total Hamiltonian $H^{'}_T$ becomes:
 \begin{equation*}
[U_{SB},H'_T]=[U_{SB},H_T]+\epsilon[U_{SB},H'\otimes\mathcal{I}_B].
\end{equation*}
The commutator $[U_{SB},H'\otimes\mathcal{I}_B]\neq0$ in general even though $[U_{SB},H_T]=0$. So the unitaries may not commute with the total perturbed Hamiltonian and hence the resulting operation is no longer a thermal operation. As the parameter $\epsilon$ is small, we can say that the resulting set of operations form \textit{approximate thermal operations}~\cite{one} and denote the set by $\mathbb{TO}_{\epsilon}$. Throughout the paper, we will assume that the system's original, unperturbed Hamiltonian has non-degenerate spectra so that the tools of non-degenerate time independent perturbation theory can be applied to describe the effect of perturbation.

We explore the concept of approximate Markovian thermal operations in the following subsection.
 \subsection{Approximate markovian thermal operations}
 The idea behind approximate Markovian thermal operations ($\mathbb{MTO}_\epsilon$) is the same as that of $\mathbb{MTO}$ i.e. they are those approximate thermal operations under which the total state of the system and the environment remains in a tensor product form throughout the evolution through applying the global unitaries,
\begin{equation*}
U_{SB}\rho^{\epsilon}_S \otimes \tau_B U_{SB}^{\dagger}=\rho'^{\epsilon}_S \otimes \tau_B.
\end{equation*}
Where
\begin{equation}
   \rho'^{\epsilon}_S=\Tr_B\left(U_{SB}\rho^{\epsilon}_S \otimes \tau_B U_{SB}^{\dagger}\right). \label{rhoep}
\end{equation}
The set $\mathbb{MTO}_\epsilon$ is a subset of $\mathbb{TO}_{\epsilon}$ in the same way as $\mathbb{MTO}$ is a subset of $\mathbb{TO}$.

In the subsequent section, we elucidate about the laws governing the possible state transformations under thermal operations, and derive the conditions for the same in case of Markovian thermal operations and approximate Markovian thermal operations.   
 \begin{figure}
 \vspace{-4.15cm}
\hspace{1cm}
\includegraphics[scale=0.4]{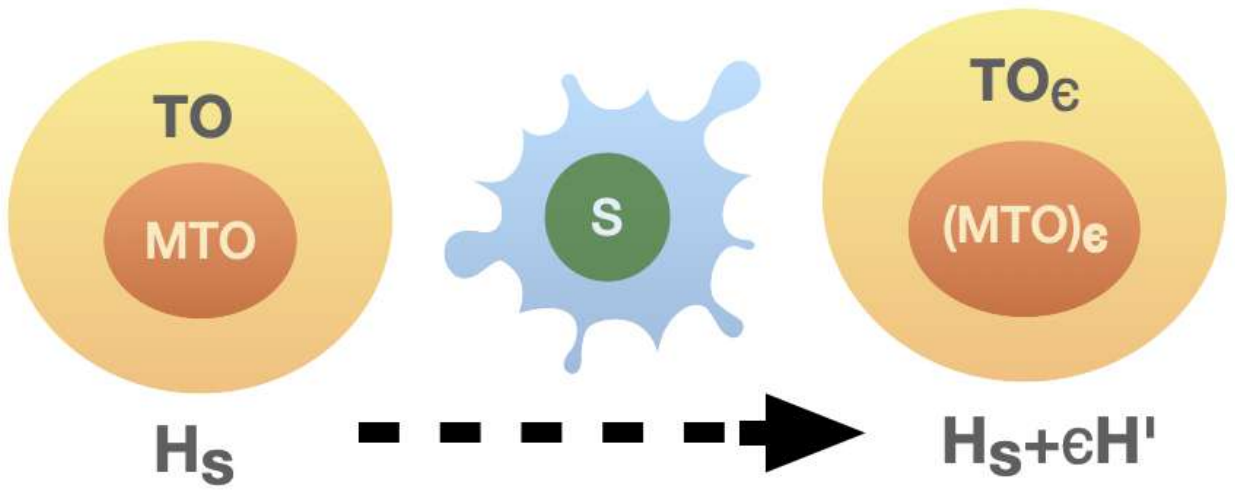}
\vspace{-3cm}
\caption{Schematic diagram showing the relationship between the sets of $\mathbb{MTO}$ and $\mathbb{TO}$. The former is a subset of the latter. In presence of perturbation in the system Hamiltonian, $\mathbb{TO}$ and $\mathbb{MTO}$ changes to the sets of $\mathbb{TO}$$_\epsilon$ and $\mathbb{MTO}_\epsilon$ respectively.  Similar to the unperturbed scenario, the set of  $\mathbb{MTO}_\epsilon$ resides inside the set of $\mathbb{TO}$$_\epsilon$.}
\label{f1}
\end{figure}

 \section{Conditions for state transformation under $\mathbb{MTO}$ and $\mathbb{MTO}_\epsilon$} \label{Conds}
Under $\mathbb{TO}$, a state transforms~\cite{PhysRevLett.115.210403} in the following manner. The
diagonal elements of a state transform as
\begin{equation}
\Lambda  \ket{i} \bra{i}= \sum_{j} P(i \rightarrow j) \ket{j} \bra{j}.
\end{equation}
 For systems whose Hamiltonian has a non-degenerate Bohr spectra i.e. the differences between any two energy levels are not degenerate, the off-diagonal elements transform as
    \begin{equation}
\Lambda  \ket{i} \bra{j}= \Lambda_{ij}\ket{i} \bra{j}, i\neq j.
\end{equation}

Here, $P(i \rightarrow j)$ denotes the transition probability from the energy level $i$ to $j$, whereas $\Lambda_{ij}$ refer to the factors by which the off-diagonal elements are damped during the evolution.

Since $\mathbb{MTO}$ is a subset of $\mathbb{TO}$, the state transformation laws are same under $\mathbb{MTO}$. However, the condition of Markovianity imposes extra constraint on the parameters of global unitaries applied on system-environment. 

Let $\{E_R\}$ denote the environment energy levels and $\{\alpha_{E_R}^{ji}\}$ refer to the parameters of the global unitary. Also, let $E_i-E_j=\omega_{ji}$ be the difference the energy levels of the system's Hamiltonian. Then, the parameters of the unitaries that generate the set of $\mathbb{MTO}$ are constrained as 
\begin{equation}
    |\alpha_{E_R}^{ji}|^{2}= \frac{P(E_R+\omega_{ji})P(i \rightarrow j)}{P(E_R)}, \label{eq:1}
\end{equation}
\begin{equation}
    \alpha_{E_R}^{ii}({\alpha_{E_R}^{jj}}^{*})= \Lambda_{ij}, \forall E_R. \label{eq:2}
\end{equation}

The proofs of equations \eqref{eq:1} and \eqref{eq:2} are detailed in the appendices \ref{appA} and \ref{appB}.

Similarly, the state transformation laws under $\mathbb{MTO}_\epsilon$ are same as that of  $\mathbb{TO}_{\epsilon}$ , as the former is just the subset of the latter. The presence of Markovianity in $\mathbb{MTO}_\epsilon$  imposes extra constraints on the parameters of global unitaries applied on initial system-environment composite state. The conditions obtained are identical to equations \eqref{eq:1} and \eqref{eq:2} (see appendix \ref{appC} for the proof).

In the next section, we give a brief account of some of the existing measures of non-Markovianity. We also provide the definitions of our proposed measures of non-Markovianity, in the context of thermal operations and approximate thermal operations. 

 \section{Measures of Non-Markovianity } \label{nm}
 The amount of non-Markovianity of a given dynamics can be quantified by any nonnegative function whose value is zero if the dynamics is Markovian. Such a function is referred to as measure of non-Markovianity. A number of measures exist in literature for the quantification of non-Markovianity. In the following subsection, we briefly explain some of the measures.
 \subsection{Brief discussion on existing measures}
 \label{A}
  In this subsection, we discuss about the RHP measure, the BLP measure and the geometric measure of non-Markovianity. 
  
  The RHP measure~\cite{PhysRevLett.105.050403,PhysRevA.83.052128} introduced by Rivas, Huelga and Plenio, is based on the concept of CP-divisibility of a dynamical map. Any completely positive trace preserving (CPTP) map $\Lambda(t,t_0)$ is called CP-divisible if it can be decomposed as,
\begin{equation*}
    \Lambda\left(t_1+\tau,t_0\right)= \Lambda\left(t_1+\tau,t_1\right) \Lambda\left(t_1,t_0\right),
\end{equation*}
with $\Lambda\left(t_1+\tau,t_1\right) $ being a completely positive map for any $t_1,\tau >0$. RHP considered any map to be Markovian iff it is CP-divisible. Based on this they defined a function,
\begin{equation*}
    g(t)= \lim_{\tau \rightarrow 0^+}\frac{ ||\Lambda(t+\tau,t)\otimes \mathcal{I}_d(\ket{\Phi}\bra{\Phi})||_1-1}{\tau}.
\end{equation*}
Here $||A||_1=\Tr{\sqrt{A^{\dagger}A}}=\Tr{\sqrt{AA^{\dagger}}}$ is the trace norm of any linear operator $A$ and $\ket{\Phi}=\frac{1}{\sqrt{d}}|\sum_{i=0}^{d-1}\ket{ii}$ is the maximally entangled state in $\mathbb{C}^d \otimes \mathbb{C}^d $ bipartite system. 

 In general $g(t)\geq 0$, with the equality holding only when the map $\Lambda(t,0)$ is CP-divisible and hence Markovian. On this notion, RHP defined a measure of non-Markovianity $\mathcal{N}_{RHP}(\Lambda)$ for any map $\Lambda$ as,
\begin{equation*}
    \mathcal{N}_{RHP}(\Lambda) = \frac{\mathscr{I}}{\mathscr{I}+1},
\end{equation*}
where $\mathscr{I}=\int_{0}^{\infty} g(t) dt$. 

 Proposed by Breuer, Laine and Piilo, the BLP measure ~\cite{PhysRevLett.103.210401,PhysRevA.83.052128} is based on the fact that the distinguishability of a pair of states evolving under Markovian maps decreases with time. The distinguishability $\mathscr{D}\left( \rho_1, \rho_2\right)$, of two states $\rho_1$ and $\rho_2$ is measured by taking the trace distance between them, 
 \begin{equation*}
     \mathscr{D}\left( \rho_1, \rho_2\right)=\frac{1}{2} ||\rho_1 - \rho_2||_1,
 \end{equation*}
 and
 \begin{equation*}
   \mathscr{R}\left(\rho_1, \rho_2;t\right)  = \frac{d}{dt} \mathscr{D}\left (\rho_1, \rho_2\right),
 \end{equation*}
 is the rate of change of distinguishability between two states $\rho_1$ and $\rho_2$. According to BLP, for Markovian process we must have $\mathscr{R} \leq 0$. Positive value of $\mathscr{R}$ indicates that the process is non-Markovian. Hence for any map $\Lambda$, they introduced a measure of non-Markovianity, $ \mathcal{N}_{BLP}(\Lambda)$ given as,
 \begin{equation*}
    \mathcal{N}_{BLP}(\Lambda) = \sup_{\rho_1, \rho_2}\int_{\mathscr{R}>0}\mathscr{R}\left(\rho_1, \rho_2;t\right).
\end{equation*}

Lastly, we consider the geometric measure of non-Markovianity~\cite{Rivas_2014}. Let $\mathcal{A}$ be the set of all CPTP maps and $\mathcal{M}$ be its subset, containing those maps that are also Markovian. Then, a geometric measure of non-Markovianity corresponding to any map belonging to the set $\mathcal{A}$ can be defined as the distance between the map and its nearest Markovian counterpart. If $\mathbb{D}(\Lambda_1,\Lambda_2)$ with $\Lambda_1,\Lambda_2 \in \mathcal{A}$ be a suitable distance measure that is defined in the space of maps, the geometric measure of non-Markovianity is given by,
  \begin{equation*}
      \mathcal{N}^{geo}_t[\Lambda(t,t_0)]\coloneqq \min_{\Lambda_{M} \in \mathcal{M}} \mathbb{D}[\Lambda(t,t_0),\Lambda_{M}(t,t_0)]. 
  \end{equation*}
In the next subsection, we propose three measures of non-Markovianity, specific to $\mathbb{TO}$ and $\mathbb{TO}_{\epsilon}$. 
 \subsection{Measures of non-Markovianity for $\mathbb{TO}$ and $\mathbb{TO}_{\epsilon}$ } \label{B}
 Let $\mathbb{X}$ be the set of all thermal operations and $\mathbb{X}_\epsilon$ denote the set of all approximate thermal operations. Elements belonging to $\mathbb{X}$ and $\mathbb{X}_\epsilon$ map a state in the input Hilbert space $\mathbb{H}_{in}$ to a state defined in the output Hilbert space $\mathbb{H}_{out}$.  Both the input and output Hilbert spaces have dimension $d_1$. Let $\mathbb{L}_S$ denotes the set of quantum states defined in $H_{in}$. $\mathbb{X}$ and $\mathbb{X}_\epsilon$ comprises of operations that can be Markovian~\cite{Son_2025,PhysRevA.105.012420} as well as non-Markovian. The set of all Markovian thermal operations forms a subset of $\mathbb{X}$ and is denoted by $\mathbb{X}^M$. Similarly, the set of all approximate Markovian thermal operations is designated as $\mathbb{X}^M_\epsilon$.
 
In general, the global unitaries $U_{SB}$ generating $\mathbb{TO}$ and $\mathbb{TO}_{\epsilon}$ can create entanglement between the system and the environment. This means that even if the system and the environment are uncorrelated before being subjected to the action of $U_{SB}$, they might become entangled after the evolution. On the other hand, the unitaries corresponding to the operations $\Lambda \in \mathbb{X}^M$ and $\Lambda_{\epsilon} \in \mathbb{X}^M_\epsilon$ ensure that once the system starts off being uncorrelated with the environment, it remains uncorrelated throughout the evolution.  This means that these unitaries are not able to create entanglement between the system and the environment. 
 
 This motivates us to consider the entanglement formed between the system and the environment as our first measure of non-Markovianity. To be specific, we use the relative entropy of entanglement~\cite{PhysRevA.57.1619,PhysRevA.56.4452,PhysRevLett.78.2275} to estimate the non-Markovianity.

  Relative entropy of entanglement of a bipartite system $\rho_{SB}$ in $\mathbb{C}^{d_1}\otimes\mathbb{C}^{d_2}$ is defined as,
\begin{equation}
     E = \min_{\sigma \in \chi_s} S(\rho_{SB}||\sigma), 
 \end{equation}
 where the set $\chi_s$ denotes the set of all separable states in $\mathbb{C}^{d_1}\otimes\mathbb{C}^{d_2}$, and $S(\rho)=-\Tr(\rho \log_2 \rho)$ denotes the von Neumann entropy of the state $\rho$. 
 
 In the present case, let us consider the interaction between the system and the environment, initially in states $\rho_S$ and $\tau_B$ respectively. The combined system-environment initial state $\rho_S \otimes \tau_B$ is now evolved via global unitaries $U_{SB}$, which generate the set $\mathbb{X}$.
 Thus, for a given operation $\Lambda \in \mathbb{X}$ and given initial environment state $\tau_B$, the maximum entanglement that the corresponding operation can generate, upon acting on the set $\mathbb{L}_S$ of system's state in $\mathbb{C}^{d_1}$ can be defined as a measure of non-Markovianity,
 \begin{equation}
     E_\Lambda \coloneqq \max_{\rho_S\in\mathbb{L}_S} \biggl[\min_{\sigma \in SEP} S\left(U_{SB}\rho_{S}\otimes \tau_BU_{SB}^\dagger||\sigma\right)\biggr]. \label{Eq6}
 \end{equation}

 Suppose that the system is subjected to a perturbation and then coupled to the environment, before the combined system-environment state is evolved via the same global unitaries $U_{SB}$. The evolution of the system alone is now described by an operation $\Lambda_\epsilon \in \mathbb{X}_\epsilon$. The measure $E_\epsilon$ in this case is defined as, 
 \begin{equation}
     E^\epsilon_\Lambda \coloneqq \max_{\rho^{\epsilon}_S\in\mathbb{L}_S} \biggl[\min_{\sigma \in SEP} S\left(U_{SB}\rho^{\epsilon}_S\otimes \tau_BU_{SB}^\dagger||\sigma\right)\biggr]. \label{Eq6.1}
 \end{equation}
 Where $\rho^{\epsilon}_S$ is the altered initial state of the system due to the perturbation. 
 
We have seen that operations belonging to the sets $\mathbb{X}^M$ and $\mathbb{X}^M_\epsilon$, when acts on initially uncorrelated system and the environment, do not produce any correlation between them. It is also known that the quantum mutual information~\cite{PhysRevA.72.032317,PhysRevA.76.032327} of any bipartite product state is always zero. So if the system and the environment starts off uncorrelated and the time evolution of the system state can be described using an element $\Lambda \in \mathbb{X}$, but $\Lambda \notin \mathbb{X}^M$, the mutual information of the evolved state would be non-zero. The argument also holds true for any $\Lambda_\epsilon \in \mathbb{X}_\epsilon$ such that $\Lambda \notin \mathbb{X}_\epsilon^M$ employed in evolving the system. Based on this, we propose the quantum mutual information between the system and the environment as our second measure for the quantification for the non-Markovianity.

Let $I_\Lambda$ and $I^\epsilon_\Lambda$ denote the quantum mutual information between the system and the environment while the system undergoes time evolution through the application of thermal operations and approximate thermal operations respectively. The environment is initialized in the thermal state $\tau_B$, for both perturbed and unperturbed cases. The initial states of the system, with or without perturbations are labeled as $\rho_S$ and and $\rho^\epsilon_S$ respectively. Then,
    \begin{equation}
    I_\Lambda=S(\rho'_S)+S(\rho'_B)-S(\rho_{SB}). \label{eq:5}
    \end{equation}
    \begin{equation}
I^{\epsilon}_\Lambda=S(\rho'^{\epsilon}_S)+S(\rho'^{\epsilon}_B)-S(\rho^{\epsilon}_{SB}).\label{eqn:5}
    \end{equation}
    Where $\rho'_S$ and $\rho'^\epsilon_S$ are given in equations~\eqref{rho} and~\eqref{rhoep}, and,
    \begin{equation}
    \begin{split}
         \rho'_B=\Tr_S\left(U_{SB} \rho_S\otimes \tau_BU^{\dagger}_{SB}\right),\\
         \rho'^\epsilon_B=\Tr_S\left(U_{SB} \rho^\epsilon_S\otimes \tau_BU^{\dagger}_{SB}\right),\\
         \rho_{SB}=U_{SB}\rho_S \otimes \tau_B U_{SB}^{\dagger},\\
         \rho^\epsilon_{SB}=U_{SB}\rho^\epsilon_S\otimes \tau_B U_{SB}^{\dagger}.
         \label{9}
    \end{split}
    \end{equation}
    As a third measure, we formulate a distance-based measure, analogous to the geometric measure of non-Markovianity as discussed in Sec~\ref{A}.
 Suppose we have a $\mathbb{TO}$ $\Lambda \in \mathbb{X}$ and a $\mathbb{MTO}$ $\Lambda^M \in \mathbb{X}^M$, both acting on a state $\rho \in \mathbb{L}_S$. Then the measure of non-Markovianity for any operation $\Lambda$ is defined as
 \begin{equation}
   D_\Lambda \coloneqq \max_{\rho} \Big[\min_{\Lambda^M} ||\Lambda(\rho)-\Lambda^{M}(\rho)||_1\Big].   \label{Eq3}
 \end{equation}
 Likewise, we can define the measure in the space of $\mathbb{TO}_{\epsilon}$, 
\begin{equation}
   D^\epsilon_\Lambda \coloneqq \max_{\rho}  \Big[\min_{\Lambda^M_\epsilon} ||\Lambda_\epsilon(\rho)-\Lambda^M_\epsilon(\rho)||_1\Big].   \label{Eq4}
 \end{equation}
 Where, $\Lambda_\epsilon \in \mathbb{X}_\epsilon$ and $\Lambda^M_\epsilon \in \mathbb{X}^M_\epsilon$ are any arbitrary approximate thermal operation and approximate Markovian thermal operation. 
 
In the next section we present three propositions corresponding to the above three defined measures. These propositions relate the non-Markovianity measures for $\mathbb{TO}_{\epsilon}$ with the corresponding measures in case of $\mathbb{TO}$.
\section{Connecting measures of non-markovianity for $\mathbb{TO}$ with that of $\mathbb{TO}_{\epsilon}$} 
\label{measure}
We first make the readers familiar with some of the initial conditions taken into account. For instance, we consider the initial state of the environment to be a thermal state corresponding to a Hamiltonian $H_B$ and inverse temperature $\beta$, i.e. $\tau_B=\frac{e^{-\beta H_B}}{\Tr(e^{-\beta H_B})}$. The dimension of the Hilbert space of the environment is $d_2$. Local Hamiltonian of the system in unperturbed scenario is $H_S$ with eigenvalues $\{E_i\}$ and eigenstates $\{\ket{i}\}$. In the perturbed case, the system Hamiltonian is $H_S'=H_S+\epsilon H'$, with eigenvalues $\{E_{i'}\}$ and and eigenstates $\{\ket{i'}\}$, with $i=0,1,2...,d_1-1$. Thus the total dimension of the combined system +environment is $d=d_1d_2$. Initial state of the unperturbed system is taken to be $\rho_S=\sum_{i,j} P_{ij} \ketbra{i}{j}$. For the perturbed system, $\ketbra{i}{j}$ is simply replaced by $\ketbra{i'}{j'}$, so that the initial state in this case is given by $\rho^\epsilon_S=\sum_{i,j} P_{ij} \ketbra{i'}{j'}$. In both situations, the system and the environment are not interacting with each other initially. Hence,  the initial combined system-environment state is just $\rho_S \otimes \tau_B $ and $\rho^\epsilon_S \otimes \tau_B $ for the unperturbed and the perturbed case respectively. After the evolution by a global unitary $U_{SB}$, the final state of the system is given by equations~\eqref{rho} and~\eqref{rhoep}, in the absence and the presence of perturbation respectively. 

Having specified the initial conditions, we are now ready to put forward three propositions, corresponding to each of the three proposed measures of non-Markovianity as discussed in Sec.~\ref{B}. These propositions connect the perturbation-induced non-Markovianity to that of the unperturbed one.
 The first proposition is stated in the following subsection.
\subsection{Entanglement-based measure}
 We present the following proposition that relates the entanglement $E_\Lambda$ generated between the system and the environment while the evolution of the system is described by a  $\mathbb{TO}$ $\Lambda$, and the entanglement generated $E^\Lambda_\epsilon$ between the perturbed system and the environment during the system's evolution via a $\mathbb{TO}$$_\epsilon$ $\Lambda_\epsilon$. The proposition is stated in the following manner.
\begin{proposition}
   Let $E_\Lambda$ and $E^{\epsilon}_\Lambda$ be the entanglement-based measures of non-Markovianity defined in equations \eqref{Eq6} and \eqref{Eq6.1}, for a given $\mathbb{TO}$ $\Lambda$ and the corresponding $\mathbb{TO}_{\epsilon}$ $\Lambda_\epsilon$ respectively. The perturbation-induced change in the amount of entanglement generated, $\Delta E_\Lambda:=E^{\epsilon}_\Lambda-E_\Lambda$ is upper bounded by a quantity $\gamma_\Lambda$, which is specific to the global unitary $U_{SB}$ generating the operations, as well as the perturbing Hamiltonian, $H'$,
    \begin{equation}
    \Delta E_\Lambda\leq \epsilon \gamma_\Lambda.
    \end{equation}
   If the initial state of the system is given as $\rho_S =\sum_{ij} P_{ij}\ketbra{i}{j}$, with $\{\ket{i}\}$ being the eigenbasis of the system's Hamiltonian, $H_S$ and $\sigma$ is any arbitrary separable joint state of the system and the environment in $\mathbb{C}^{d_1} \otimes \mathbb{C}^{d_2} $, the parameter $\gamma_\Lambda$ is given by
   \begin{eqnarray}
\max_{P_{ij}}\min_{\sigma \in SEP}X_\Lambda=\gamma_\Lambda,
   \end{eqnarray}
   with 
   \begin{eqnarray}
   \begin{split}
X_\Lambda&=\Tr\biggl[U_{SB}\biggl(\sum_{ij} P_{ij}\Big(\sum_{k\neq i}\frac{\bra{k}H'\ket{i}}{E_i-E_k}\ketbra{k}{j}\\&+\sum_{l \neq j}\frac{\bra{j}H'\ket{l}}{E_i-E_k}\ketbra{i}{l}\Big)\otimes \tau_B\biggr)U_{SB}^{\dagger}\\&\biggl(\mathcal{I}_d+\log_2U_{SB}\sum_{ij} P_{ij}\ketbra{i}{j}\otimes \tau_B U_{SB}^{\dagger}\biggr)\biggl].
\end{split}
   \end{eqnarray}
    \end{proposition}
\begin{proof}
  The relative entropy between two states $\rho_1$ and $\rho_2$ can be expressed as, 
\begin{equation}
    S(\rho_1||\rho_2)=\Tr(\rho_1 \log_2 \rho_1-\rho_1 \log_2 \rho_2). \label{relent}
\end{equation}
  Using equation \eqref{relent}, we can rewrite equations \eqref{Eq6} and \eqref{Eq6.1} as,
  \begin{equation}
E_\Lambda=\max_{\rho_S}\biggl[\min_{\sigma \in SEP}\Tr\Big(\rho_{SB}
\log_2\rho_{SB}-\rho_{SB}\log_2 \sigma\Big)\biggr], \label{e1}
  \end{equation}
\begin{equation}  E_\Lambda^\epsilon=\max_{\rho_S^\epsilon}\biggl[\min_{\sigma \in SEP}\Tr\Big(\rho_{SB}^\epsilon
\log_2\rho_{SB}^\epsilon-\rho_{SB}^\epsilon\log_2 \sigma\Big)\biggr],\label{e2}
\end{equation}
where $\rho_{SB}$ and $\rho_{SB}^\epsilon$ are given in equation~\eqref{9}. 

  Using the form of $\rho_S$ in equation~\eqref{e1}, we have,
  \begin{equation}
      E_\Lambda=\max_{P_{ij}} \min_{\sigma \in SEP}\Tr\left(A \log_2 A-A\log_2 \sigma\right), \label{e3}
  \end{equation}
  with $A=U_{SB}\left(\sum_{ij} P_{ij}\ketbra{i}{j}\otimes \tau_B\right) U_{SB}^{\dagger}$.
  
  Also, using the form of $\rho^\epsilon_S$ and employing first order perturbation theory to relate the perturbed eigenstates $\{\ket{i'}\}$ to the unperturbed eigenstates $\{\ket{i}\}$, we can write the equation~\eqref{e2} as,
  \begin{eqnarray*}
  \hspace{-0.2cm}
  \begin{split}
     E^{\epsilon}_\Lambda&=    
   \max_{P_{ij}} \min_{\sigma \in SEP}\biggl[\Tr\Big[\Big(A+\epsilon B\Big)\\& \Big(\log_2(A+\epsilon B)-\log_2\sigma\Big)\Big]\biggr]. 
  \end{split}
  \end{eqnarray*}
  Where,  
  \begin{equation*}
  \begin{split}
B=U_{SB}\biggl(\sum_{ij} P_{ij}\biggl(\sum_{k\neq i}\frac{\bra{k}H'\ket{i}}{E_i-E_k}\ketbra{k}{j}\\+\sum_{l \neq j}\frac{\bra{j}H'\ket{l}}{E_i-E_k}\ketbra{i}{l}\biggr)\otimes \tau_B\biggr)U_{SB}^{\dagger}.
\end{split}
\end{equation*}
Note that considering weak perturbation, we have neglected the higher order terms in $\epsilon$ and kept terms up to first order in $\epsilon$.
Maximization over $P_{ij}$ is such that $\rho_S$ belongs to the set of quantum states defined in Hilbert space, $\mathbb{H}_S$ of the system.
 Now, for any Hermitian matrix $A$ and small perturbation $\epsilon$, one can write, 
 \begin{eqnarray}
 \begin{split}
    \Tr\Big[\Big(A+\epsilon B\Big) \Big(\log_2(A+\epsilon B)\Big)\Big]\\&\hspace{-5cm}\approx\Tr\Big[A\log_2A\Big]+\epsilon \Tr\Big[B\left(\mathcal{I}_d+\log_2A\right)\Big].
    \end{split}
 \end{eqnarray}
Here, $d=d_1d_2$ is the total dimension of the system + environment. Thus one can write,
\begin{eqnarray*}
\begin{split}
   E^{\epsilon}_\Lambda= \max_{P_{ij}} \min_{\sigma \in SEP}\left[\Tr\left(A\log_2\dfrac{A}{\sigma}\right)+\epsilon X_\Lambda\right], 
\end{split}
\end{eqnarray*}
where we have $X_\Lambda=\Tr\left[B\left(\mathcal{I}_d+\log_2\dfrac{A}{\sigma}\right)\right]$ and $\log_2 \dfrac{A}{\sigma}=\log_2 A-\log_2\sigma$.

 Using the inequalities
\begin{equation*}
\begin{split}
    \min_x\Big(f(x)+g(x)\Big)\leq \min_x f(x)+\max_x g(x),\\
    \max_x\biggr(f(x)+g(x)\Big)\leq \max_x f(x)+\max_x g(x).
\end{split}
\end{equation*}
 We get,
 \begin{equation}
 \begin{split}
E^{\epsilon}_\Lambda\leq \max_{P_{ij}} \min_{\sigma \in SEP}\biggl[\Tr\left(A \log_2 A-A\log_2 \sigma\right)\biggr]\\+\epsilon\max_{P_{ij}} \max_{\sigma \in SEP} X_\Lambda. \label{e5}
 \end{split}
  \end{equation} 
   The first term of the inequality~\eqref{e5} is just $E_\Lambda$ (equation~\eqref{e3}). Labelling,
   \begin{equation*}
\max_{P_{ij}}\min_{\sigma \in SEP} X_\Lambda=\gamma_\Lambda,
   \end{equation*} 
   We have,
\begin{equation}
  E^{\epsilon}_\Lambda-E_\Lambda\leq\epsilon \gamma_\Lambda.  
\end{equation}
This completes the proof.
\end{proof}
Thus we derived a bound on perturbation induced change in the amount of total entanglement generated between the system and the environment, considering the relative entropy of entanglement, maximized over the system's states. Next we ask that, can perturbation lead to more entanglement generation? In other words, can initial perturbations enhance non-Markovianity? 
 Below, we present a numerical example where we answer this question affirmatively. Unlike the proposition, for numerical ease we consider logarithmic negativity~\cite{PhysRevA.65.032314,PhysRevLett.95.090503} as a measure of entanglement in the following example.
 
 The logarithmic negativity for a bipartite system $AB$ in a state $\rho$ is defined as,
 \begin{equation}
E_N(\rho)\coloneqq\log_2||\rho^{\Gamma_A}||_1,
 \end{equation}
 where $\Gamma_A$ denotes partial transpose with respect to the subsystem $A$. \\\\
 
\noindent \textbf{Example.} We consider $\mathbb{C}^2\otimes \mathbb{C}^2$ joint states, which means that the system as well as the environment is taken as a qubit. Local Hamiltonian of the system is $H_S=\sigma_z \delta$, where $\delta$ is the unit of energy. The system is initially considered to be diagonal in energy eigenbasis, as $\rho_S=a\ketbra{0}{0}+(1-a)\ketbra{1}{1}$, where the set $\{\ket{0},\ket{1}\}$ contains the eigenstates of $\sigma_z$. The system is then perturbed by an external perturbing Hamiltonian $H'=\sigma_x \delta$. Perturbed Hamiltonian of the system is thus given by $H_S'=\left(\sigma_z+\epsilon\sigma_x \right)\delta$. So, the initial state of the perturbed system is $\rho^{\epsilon}_S=a\ketbra{0'}{0'}+(1-a)\ketbra{1'}{1'}$, where $\ket{0'}$ and $\ket{1'}$ are eigenstates of the perturbed system Hamiltonian $H'_S$. Local Hamiltonian of the environment is $H_B=\sigma_z \delta$ and the initial state of environment is considered to be the thermal state $\tau_B=\frac{e^{-\beta \sigma_z}}{\Tr(e^{-\beta \sigma_z})}$, with the inverse temperature $\beta=\dfrac{1}{T} \dfrac{k_B}{\hslash\delta}$, where $\dfrac{\hslash\delta}{k_B}$ gives the unit of temperature. Here, $k_B$ is the Boltzmann constant and $\hslash$ is the reduced Planck's constant.

Now if $U_{SB}$ is the global unitary generating the $\mathbb{TO}$, then $\big[U_{SB}, H_S\otimes \mathcal{I}_2+\mathcal{I}_2\otimes H_B]=0$. This suggests that $U_{SB}$ must be diagonal in the energy eigenbasis of the total Hamiltonian, $H_T=\left(H_S\otimes \mathcal{I}_2+\mathcal{I}_2\otimes H_B\right)\delta$. However in appendix~\ref{appD}, we show that for the concerned example, unitaries of the type $U_{SB}=\sum_{i,j}e^{-\mathbb{I}\alpha_{ij}}\ketbra{ij}{ij}$ with $i,j=0,1$ and $\mathbb{I}=\sqrt{-1}$ can never generate entanglement. 

In our example, the total Hamiltonian $H_T$ has degenerate spectra, with $\ket{01}$ and $\ket{10}$ having same energies. Thus any linear combination of  $\ket{01}$ and $\ket{10}$ must be an eigenstate of $U_{SB}$, to satisfy the energy conservation. Hence we choose global unitary of the type,
\begin{equation*}
\begin{split}
    U_{SB}=e^{-\mathbb{I}\alpha_1}\ket{00}\bra{00}+e^{-\mathbb{I}\alpha_2}\ket{11}\bra{11}\\+e^{-\mathbb{I}\alpha_3}\ket{\psi_1}\bra{\psi_1}+e^{-\mathbb{I}\alpha_4}\ket{\psi_2}\bra{\psi_2}.
\end{split}
\end{equation*}
with,
\begin{equation*}
\begin{split}
   \ket{\psi_1}=\sqrt{\frac{2}{3}}\ket{01}+\sqrt{\frac{1}{3}}\ket{10},\\
   \ket{\psi_2}=\sqrt{\frac{1}{3}}\ket{01}-\sqrt{\frac{2}{3}}\ket{10}.
\end{split}
\end{equation*}
Clearly, $[U_{SB},H_T]=0$ and $U_{SB}$ being no longer diagonal in the computational basis, is now able to generate entanglement.

The amount of entanglement $E_N$ produced between the system and the environment, when the system is subjected to a $\mathbb{TO}$ is given by,
\begin{equation*}
E_N=\log_2||(\rho_{SB})^{\Gamma_A}||_1,
\end{equation*}
 Likewise the amount of entanglement generated while the system is subjected to the corresponding $\mathbb{TO}$$_\epsilon$ is given by,
\begin{equation*}
E_N^\epsilon=\log_2||(\rho_{SB}^\epsilon)^{\Gamma_A}||_1,
\end{equation*}
The form of $\rho_{SB}$ and $\rho_{SB}^{\epsilon}$ are given in equation~\eqref{9}.

So the perturbation-induced change in entanglement is given by,
\begin{equation*}
    \Delta E_N=E_N^\epsilon-E_N.
\end{equation*}
\begin{figure}
\vspace{-0.5cm}
\includegraphics[scale=0.25]{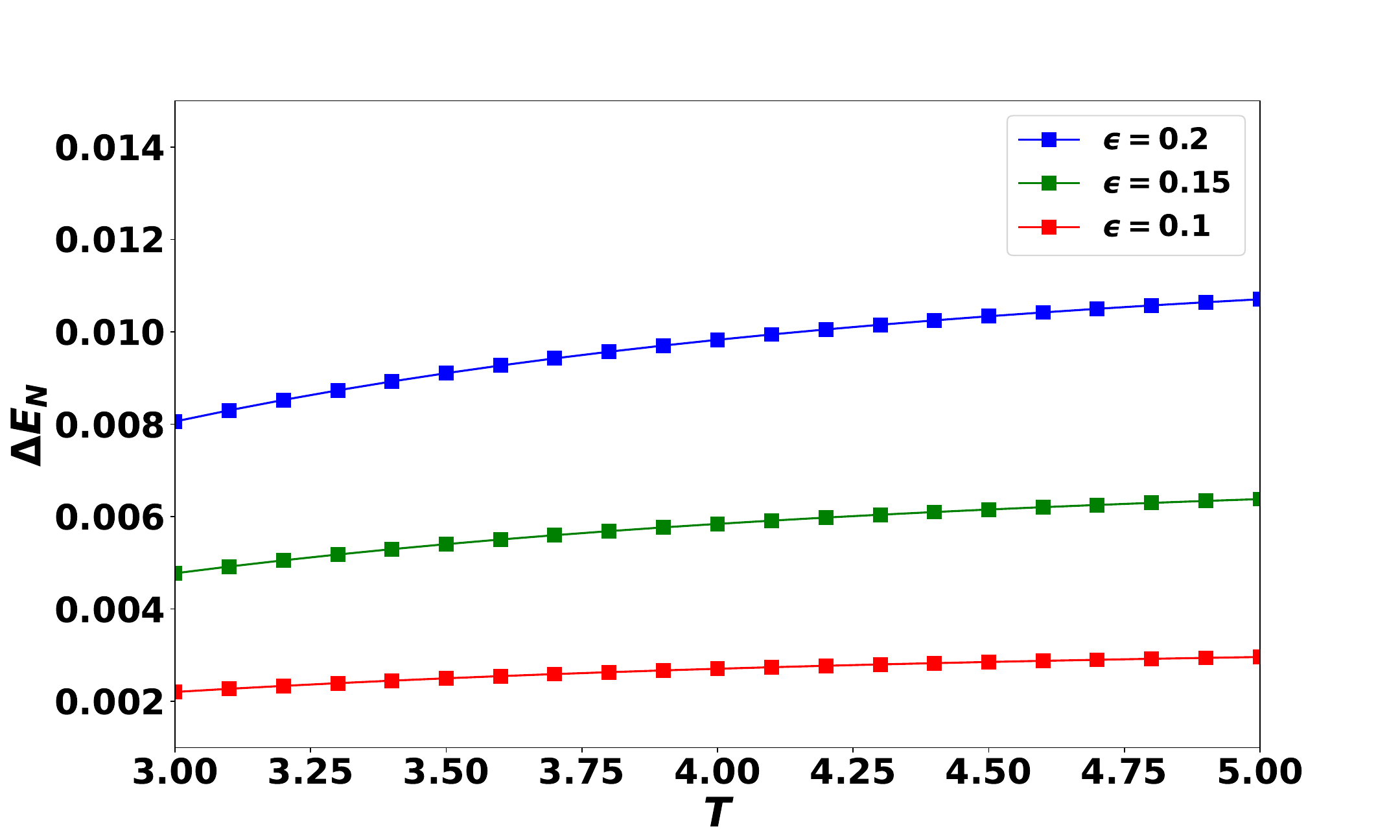}
\caption{Plot depicting the variation of the difference in the amount of entanglement generated $\Delta E_N=E^{\epsilon}_N-E_N$  with the environment temperature $T$, for $\mathbb{TO}$$_\epsilon$ and $\mathbb{TO}$, for three values of perturbation parameters, $\epsilon=$ 0.1, 0.15 and 0.2. It can be seen that for all three values of $\epsilon$, $\Delta E_N$ is positive and it increases with increasing $T$. This suggest that in this case, the perturbation enhances the non-Markovianity. The vertical axis is in the units of e-bits, while the horizontal axis is in units of $\dfrac{\hslash\delta}{k_B}$ ($\delta$ is the unit of energy).} 
\label{fig6}
\end{figure}
In fig.~\ref{fig6} we plot $ \Delta E_N$ with respect to the environment temperature $T$, by fixing the parameters
 $\alpha_1 = 10^5, \alpha_2 = 2\times10^5,\alpha_3 = 3\times10^5,\alpha_4 = 4\times10^5$ and $a=0.9$, for three different sets of degrees of perturbation which are $\epsilon=0.1,\epsilon=0.15$ and $\epsilon=0.2$. We see from the plot that $\Delta E_N$ is positive and it increases with increasing $T\in[3,5]$ units and $\epsilon$. This particular example thus illustrates that the presence of perturbation in the system can give rise to an enhancement in the non-Markovianity.
 
The second proposition relating the measure of non-Markovianity in terms of the total correlation produced between the system and the environment, due to the application of $\mathbb{TO}$ and $\mathbb{TO}$$_\epsilon$ is presented in the next subsection.

\subsection{Total correlation-based measure}
Here we give our second proposition relating the correlation-based non-Markovianity measures $I_\Lambda$ and $I^\epsilon_\Lambda$. The proposition is stated as follows.
\begin{proposition} \label{prop2}
  Let $I_\Lambda$ be the total correlation generated between the system and the environment due to a thermal operation $\Lambda$, as given in equation~\eqref{eq:5} and $I^\epsilon_\Lambda$ be the total correlation created between the system and the environment by the corresponding approximate thermal operation $\Lambda_\epsilon$, as given in equation~\eqref{eqn:5}. Then the difference $\Delta I_\Lambda:= I^\epsilon_\Lambda-I_\Lambda$ is given as,
   \begin{equation}
\Delta I_\Lambda=\epsilon\theta_\Lambda.
   \end{equation}
   Where $\theta_\Lambda$ given in equation~\eqref{th} is a quantity that depends on the unitary generating $\Lambda$ and the perturbing Hamiltonian $H'$.
\end{proposition}
\begin{proof}
Writing the initial perturbed state of the system $\rho^{\epsilon}_S$ in terms of unperturbed basis states, we have
    \begin{eqnarray*}
    \begin{split}
\rho^{\epsilon}_S&=\sum_{ij}P_{ij}\biggl[\ketbra{i}{j} +\epsilon\biggl(\sum_{k\neq i}\frac{\bra{k}H'\ket{i}}{E_i-E_k}\ketbra{k}{j}\\&+\sum_{l \neq j}\frac{\bra{j}H'\ket{l}}{E_i-E_k}\ketbra{i}{l}\biggr)+O\left(\epsilon^{2}\right)\biggr].
    \end{split}
    \end{eqnarray*}
Again considering weak perturbation we omit the higher order terms in $\epsilon$ and keep terms up to first order in $\epsilon$.
    Writing 
    \begin{equation*}
        \tilde{\rho}=\sum_{k\neq i}\frac{\bra{k}H'\ket{i}}{E_i-E_k}\ketbra{k}{j}+\sum_{l \neq j}\frac{\bra{j}H'\ket{l}}{E_i-E_k}\ketbra{i}{l},
    \end{equation*}
    we get $\rho^{\epsilon}_S$ in terms of $\rho_S$ as
    \begin{equation*}
\rho^{\epsilon}_S=\rho_S+\epsilon \tilde{\rho}.
    \end{equation*}
The evolved state of the unperturbed system upon the the action of the $\mathbb{TO}$ is given as $\rho_S'=\Lambda(\rho_S)$.
   In the perturbed scenario, the evolved state of the system upon the action of the corresponding $\mathbb{TO}$$_\epsilon$ $\Lambda_\epsilon$ is given as,
    \begin{eqnarray*}
    \begin{split}
\rho'^{\epsilon}_S&=\Lambda_\epsilon( \rho^{\epsilon}_S)\\
&=\Tr_B\left(U_{SB}\rho^{\epsilon}_S \otimes \tau_B U_{SB}^{\dagger}\right)\\
&=\Tr_B\left[U_{SB}\left(\rho_S+\epsilon \tilde{\rho}\right)\otimes\tau_B U_{SB}^{\dagger}\right]\\
&=\Tr_B\left(U_{SB}\rho_S\otimes\tau_B U_{SB}^{\dagger}\right)\\&+\epsilon \Tr_B\left(U_{SB}\tilde{\rho}\otimes\tau_B U_{SB}^{\dagger}\right).
       \end{split}
    \end{eqnarray*}
  So, we can write $\rho'^{\epsilon}_S=\rho'_S+\epsilon \beta_1$ with 
  $ \beta_1=\Tr_B\left(U_{SB}\tilde{\rho}\otimes\tau_B U_{SB}^{\dagger}\right)$.

  Now we calculate the Von-Neumann entropy of $\rho'^{\epsilon}_S$:
  \begin{align*}
      S(\rho'^{\epsilon}_S)&=-\Tr\left[\left(\rho'_S+\epsilon \beta_1\right)\log_2\left(\rho'_S+\epsilon \beta_1\right)\right],\nonumber\\
      &\approx-\Tr\left(\rho'_S \log_2 \rho'_S\right)-\epsilon \Tr\left[\beta_1\left(\mathcal{I}_{d_1}+\log_2 \rho'_S\right)\right].
      \end{align*}

  Thus we can express the Von-Neumann entropy of $\rho'^{\epsilon}_S$ in terms of that of $\rho'_S$ as
  \begin{equation*}
S(\rho'^{\epsilon}_S)=S(\rho'_S)-\epsilon \bar{A},
  \end{equation*}
  with $\bar{A}=\Tr\left[\beta_1\left(\mathcal{I}_{d_1}+\log_2 \rho'_S\right)\right]$.

 Proceeding in the same way, we obtain the expressions connecting the Von-Neumann entropies of the perturbed states $\rho'^{\epsilon}_B$ and $\rho^{\epsilon}_{SB}$ with the corresponding unperturbed states $\rho'_B$ and $\rho_{SB}$,
  \begin{equation*}
  \begin{split}
    S(\rho'^{\epsilon}_B)=S(\rho'_B)-\epsilon \bar{B},\\
S(\rho^{\epsilon}_{SB})=S(\rho_{SB})-\epsilon \bar{C}, 
  \end{split}
  \end{equation*}
  with 
  \begin{equation*}
      \begin{split}
          \bar{B}=\Tr\left[\beta_2\left(\mathcal{I}_{d_2}+\log_2 \rho'_B\right)\right],\\    \bar{C}=\Tr\left[\beta_3\left(\mathcal{I}_{d}+\log_2 \rho_{SB}\right)\right].
      \end{split}
  \end{equation*}
  Where $\beta_2$ and $\beta_3$ are given by,
  \begin{equation*}
  \begin{split}
\beta_2=\Tr_S\left(U_{SB}\tilde{\rho}\otimes\tau_B U_{SB}^{\dagger}\right),  \\ 
\beta_3=U_{SB}\tilde{\rho}\otimes\tau_B U_{SB}^{\dagger}.    
  \end{split}
\end{equation*}
  Collecting all terms, the expression for mutual information in presence of perturbation can be written as,
  \begin{align*}
I^{\epsilon}_\Lambda&=S(\rho'^{\epsilon}_S)+S(\rho'^{\epsilon}_B)-S(\rho^{\epsilon}_{SB}),\\
    &=S(\rho'_S)+S(\rho'_B)-S(\rho_{SB})+\epsilon\left(\bar{C}-\bar{A}-\bar{B}\right).
  \end{align*}
  Using equation~\eqref{eq:5} and labelling,
  \begin{equation}
      \theta_\Lambda=\left(\bar{C}-\bar{A}-\bar{B}\right),
      \label{th}
  \end{equation}
 We arrive at the required result
 \begin{equation}
I^{\epsilon}_\Lambda=I_\Lambda+\epsilon\theta_\Lambda.
 \end{equation} 
 This completes the proof.
\end{proof}
Note that $\theta_\Lambda$ can be positive or negative, indicating that the total correlation may increase or decrease in presence of perturbation. However we are interested in a scenario when perturbation leads to an increase in non-Markovianity.
Below, we give a numerical example of such a scenario when the total correlation increases in presence of perturbation.\\

\noindent \textbf{Example.} We consider the combined state of system-environment in $\mathbb{C}^2\otimes\mathbb{C}^3$, i.e. the system is taken to be a qubit and the environment as a qutrit. The unperturbed system Hamiltonian is given by $H_S=\omega\sigma_z \delta$, with an eigenbasis $\{\ket{0},\ket{1}\}$. The initial state of the system, in absence of perturbations, is given in the eigenbasis of $H_S$ as $\rho_S=a\ketbra{0}{0}+(1-a)\ketbra{1}{1}$. The environment Hamiltonian is given by one of the Gell-Mann matrices,
\begin{equation*}
    H_B=\begin{pmatrix}
0 & 1 & 0\\
1 & 0 & 0\\
0&0&0
\end{pmatrix}\delta.
\end{equation*}
having an eigenbasis $\{\ket{\bar0},\ket{\bar1},\ket{\bar2}\}$. Initially the environment is considered to be in the thermal state, $\tau_B=\frac{e^{-\beta H_B}}{\Tr(e^{-\beta H_B})}$, with the inverse temperature $\beta=\dfrac{1}{T} \dfrac{k_B}{\hslash\delta}$.

Note that the total Hamiltonian $H_T=H_S \otimes \mathcal{I}_3+ \mathcal{I}_2 \otimes H_B$ has a non-degenerate eigenspectra. Since the global unitary $U_{SB}$ generating $\mathbb{TO}$ $\Lambda$ must obey the commutation relation $[U_{SB}, H_T]=0$, we choose $U_{SB}$ such that it is diagonal in the eigenbasis of $H_T$,
\begin{equation*}
\begin{split}
    U_{SB}=e^{-\mathbb{I}\alpha_1}\ketbra{0\bar{0}}{0\bar{0}}+e^{-\mathbb{I}\alpha_2}\ketbra{0\bar{1}}{0\bar{1}}+e^{-\mathbb{I}\alpha_3}\ketbra{0\bar{2}}{0\bar{2}}\\+e^{-\mathbb{I}\alpha_4}\ketbra{1\bar{0}}{1\bar{0}}+e^{-\mathbb{I}\alpha_5}\ketbra{1\bar{1}}{1\bar{1}}+e^{-\mathbb{I}\alpha_6}\ketbra{1\bar{2}}{1\bar{2}}.
\end{split}
\end{equation*}
Under perturbation, system's new Hamiltonian becomes $H'_S=\left(\sigma_z+\epsilon \sigma_x\right)\delta$. The modified initial state is now given as $\rho^{\epsilon}_S=a\ketbra{0'}{0'}+(1-a)\ketbra{1'}{1'}$. Where $\{\ket{0'},\ket{1'}\}$ is the eigenbasis of $H'_S$. Thus the total
correlation in unperturbed and the perturbed scenario can be numerically computed using equations~\eqref{eq:5} and~\eqref{eqn:5}. So the  perturbation-induced change in the total correlation is given by $\Delta I=I_\epsilon-I$.

We fix
the parameters of $U_{SB}$ as $\alpha_1=18\times 10^7,\alpha_2=30\times 10^7,\alpha_3=60\times 10^7,\alpha_4=80\times 10^7,\alpha_5=70\times 10^7,\alpha_6=90\times 10^7$, and take $a=0.9$. For this particular example, we plot $\Delta I$ vs environment temperature $T$ in fig.~\ref{fig3}. As it can be seen in the plot that $\Delta I$ is always positive for the concerned range of $T\in[0,1]$ units and it increases as $T$ decreases. The the degree of perturbation is taken as $\epsilon=0.2$.
 Thus for this example, we get an enhancement of non-Markovianity, induced by perturbation. 
\begin{figure}
\vspace{-0.5cm}
  \includegraphics[scale=0.25]{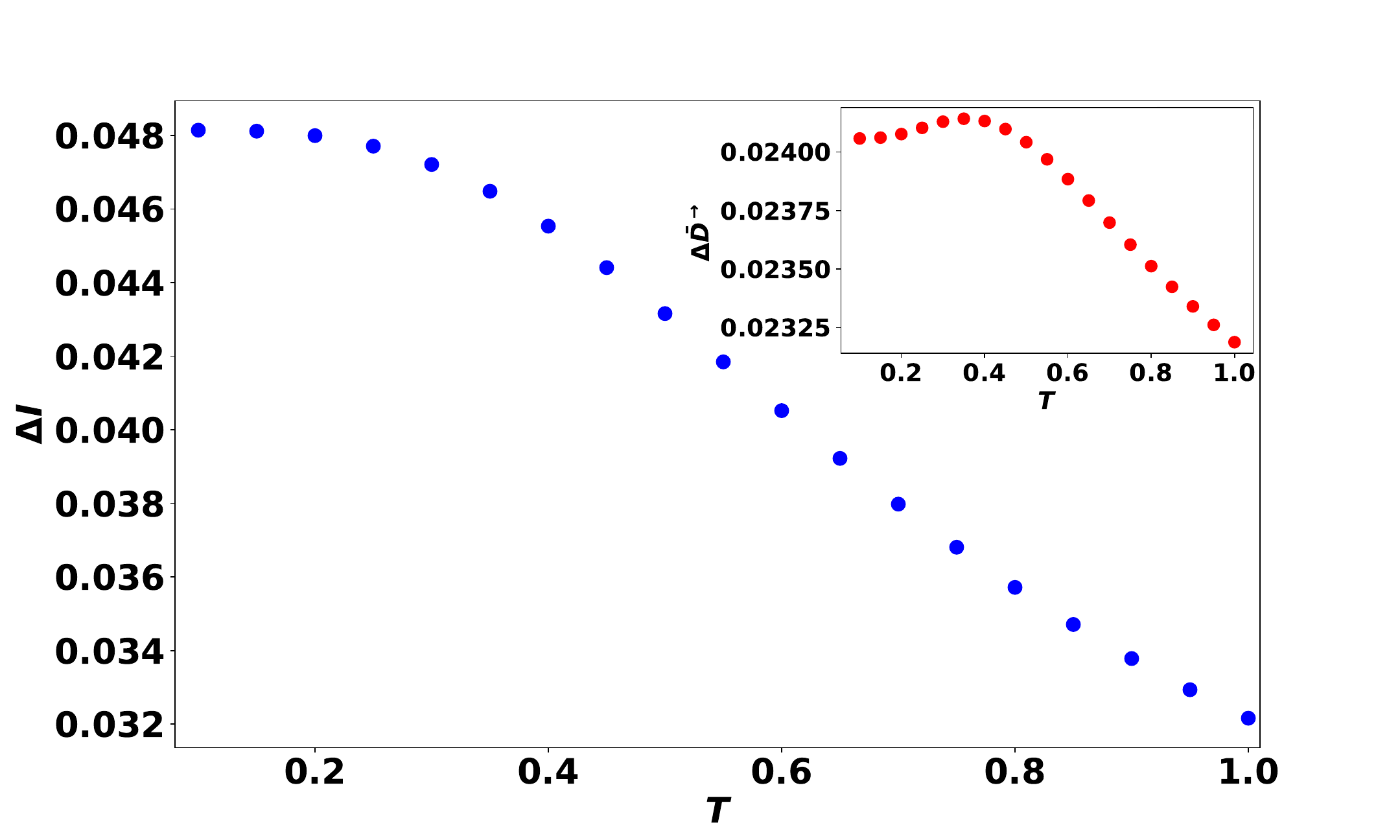} %
\caption{The outer plot depicts the behavior of the difference in the amount of total correlation  $\Delta I=I_{\epsilon}-I$ with the environment temperature $T$, for $\mathbb{TO}$$_\epsilon$ and $\mathbb{TO}$. The inset shows the variation of the difference in quantum discord $\Delta \bar{D}^{\rightarrow}=\bar{D}_{\epsilon}^{\rightarrow}-\bar{D}^{\rightarrow}$, due to $\mathbb{TO}$$_\epsilon$ and $\mathbb{TO}$, with the environment temperature. The perturbation parameter is fixed at $\epsilon=0.2$. From both the plots, it is seen that for all values of $T$ in the concerned range, $\Delta I_N$ and $\Delta \bar{D}^{\rightarrow}$ is positive. This suggests that for the chosen example perturbation enhances total correlation as well as the quantum correlation. The vertical axis is in the units of bits, while the horizontal axis is in units of $\dfrac{\hslash\delta}{k_B}$($\delta$ is the unit of energy).} 
\label{fig3}
\end{figure}
In appendix~\ref{appD}, we show  that the form of the unitary chosen in this example cannot generate entanglement. Thus, we look for quantum correlations other than entanglement which could have been generated due to the perturbation introduced for this particular example. In particular, we calculate the quantum discord that quantifies the total quantum correlation~\cite{PhysRevLett.108.150403} between two systems. The quantum discord is calculated by subtracting the total amount of classical correlations from the total correlation, classical+quantum, present between two systems.

The Henderson-Vedral classical correlation \cite{Henderson_2001} between the subsystems $A$ and $B$ of a bipartite system $AB$ in $\mathbb{C}^{d_1}\otimes \mathbb{C}^{d_2}$ is of the following form,
\begin{equation} C_A\left(\rho_{AB}\right)=\max_{P^{\dagger}_i P_i} \left(S(\rho_B)-\sum_{i} p_i S(\rho^{i}_B)\right), \label{eq:6}
\end{equation} 
where $\rho_{AB},\rho_{A}$ and $\rho_{B}$ are the total state of the system $AB$ and reduced states of the subsystems $A$ and $B$ respectively. The maximization is performed over all POVMs $F_i={P^{\dagger}_i P_i}$, done on system $A$. The probabilities are given by
\begin{equation*}
    p_i=\Tr\left(P_i\otimes \mathcal{I}_{d_2}\rho_{AB}P^{\dagger}_i\otimes \mathcal{I}_{d_2} \right).
\end{equation*}
The states $\rho^{i}_B$ are given by
\begin{equation*}
\rho^{i}_B=\frac{\Tr_A\left(P_i\otimes \mathcal{I}_{d_2} \rho_{AB}P^{\dagger}_i\otimes \mathcal{I}_{d_2} \right)}{\Tr\left(P_i\otimes \mathcal{I}_{d_2} \rho_{AB}P^{\dagger}_i\otimes \mathcal{I}_{d_2} \right)}.
\end{equation*}
As the first term of equation~\eqref{eq:6} does not involve any measurement, it can be re-framed as
\begin{equation}
    C_A\left(\rho_{AB}\right)= S(\rho_B)-\min_{{P^{\dagger}_i P_i}}\sum_{i} p_i S(\rho^{i}_B). \label{eq:8}
\end{equation}
Now the quantum discord is given by
\begin{equation}
     \bar{D}^{\rightarrow}=I-C_A\left(\rho_{AB}\right). \label{eq:9}
\end{equation}
Where $I$ is the total correlation.
We have used $\bar{D}^{\rightarrow}$ to indicate that the measurement has been performed in the first subsystem, as the discord is not symmetric with respect to the measurement in two subsystems in general.
Substituting equations~\eqref{eq:5} and~\eqref{eq:8} in equation~\eqref{eq:9}, we get the following formula for quantum discord:
\begin{equation}
    \bar{D}^{\rightarrow}=S(\rho_A)-S(\rho_{AB})+ \min_{{P^{\dagger}_i P_i}}\sum_{i} p_i S(\rho^{i}_B). \label{eq:10}
\end{equation}
For the particular example as discussed above, we analyze the effect of perturbation on discord.  We consider the system $S$ and the environment $B$ as the first and second subsystems $A$ and $B$ in equation~\eqref{eq:6}. For simplicity, we perform the minimization over all projective measurements, instead of POVM, on the system. In case of qubits, without loss of generality, we can consider the projective measurement elements, $P_1=\ketbra{\psi}{\psi},P_2=\ketbra{\psi^{\perp}}{\psi^{\perp}}$,
where $\ket{\psi}=\cos\frac{\theta}{2}\ket{0}+e^{i\phi}\sin\frac{\theta}{2}\ket{1},\ket{\psi^{\perp}}=-\sin\frac{\theta}{2}\ket{0}+e^{-i\phi}\cos\frac{\theta}{2}\ket{1}$, with $\theta\in[0,\pi]$ and $\phi \in[0,2\pi]$. Minimizing over the parameters $\theta, \phi$ enables us to calculate the quantum discord.

The total system-environment state $\rho_{SB}$, $\rho^\epsilon_{SB}$ and the state of the system $\rho_S$ and $\rho^\epsilon_S$, to be used in the equation~\eqref{eq:10} for calculating the discord are given in equations~\eqref{9} and ~\eqref{rho}, ~\eqref{rhoep} respectively.  The difference in discord $\Delta \bar{D}^{\rightarrow}=\bar{D}_{\epsilon}^{\rightarrow}-\bar{D}^{\rightarrow}$ is plotted with environment temperature $T$, as can be found in inset of fig.~\ref{fig3}. It is clear from the figure that for the particular example chosen the difference $\Delta \bar{D}^{\rightarrow}$ is positive with varying $T$, indicating that perturbation have enhanced quantum correlation between the system and the environment during the evolution. 

In the next subsection we present the third proposition that relates the distance-based measure of non-Markovianity for $\mathbb{TO}$$_\epsilon$ with that of the $\mathbb{TO}$.
\subsection{Relating the distance-based measures of non-Markovianity}
We relate the distance-based measure of non-Markovianity, $D_\Lambda$ and $D^\epsilon_\Lambda$, for a $\mathbb{TO}$ $\Lambda$ and $\mathbb{TO}$$_\epsilon$ $\Lambda_\epsilon$ respectively using the following proposition. The proposition is articulated as follows.
\begin{proposition}
   Let $D_\Lambda$ and $D^{\epsilon}_\Lambda$ be the distance-based measures of non-Markovianity given in equations~\eqref{Eq3} and~\eqref{Eq4}, defined in the space of thermal operations and the space of approximate thermal operations respectively. Then, the difference $\Delta D_\Lambda:=D^{\epsilon}_\Lambda-D_\Lambda$ provides an estimate of the change in the amount of non-Markovianity, caused due to the introduction of perturbations in the system and is upper bounded as,
    \begin{equation}
       \Delta D_\Lambda\leq \frac{\epsilon}{d_1} \max_{\Lambda^M}\left\Vert\chi_\Lambda\right\Vert_1,
    \end{equation}
    where we have,
     \begin{eqnarray}
     \begin{split}
          \chi_\Lambda&= \biggl[\biggl(\Lambda\otimes \mathcal{I}_{d_{1}}-\Lambda^M\otimes\mathcal{I}_{d_{1}}\biggr)\vartheta(H')\biggl],
          \end{split}
      \end{eqnarray}
with
\begin{eqnarray}
    \begin{split}
\vartheta(H')&= \biggl(\sum_{i,j,k\neq i}\frac{\bra{k}H'\ket{i}}{E_i-E_k}\ketbra{ki}{jj}\\&\hspace{-0.4cm}+\sum_{i,j,l\neq j}\frac{\bra{j}H'\ket{l}}{E_j-E_l}\ketbra{ii}{lj}\biggr),
    \end{split}
\end{eqnarray}
that intrinsically depends on the external field $H'$ perturbing the system.
      \end{proposition}
\begin{proof}
Definition of $D_\Lambda$ and $D^\epsilon_\Lambda$ requires two optimizations. The first one is the minimization over the set of $\mathbb{X}^M$ and $\mathbb{X}^M_\epsilon$ respectively, while the second one involves a maximization over the space of input states of the system. Note that the order of these optimization is irrelevant. Hence, we choose to tackle the optimization over the space of input states first. 

As any quantum map can be characterized by its corresponding Choi state~\cite{JAMIOLKOWSKI1972275}, we choose to work in Choi state formalism. We note that a similar approach was taken in ref.~\cite{Das_2021}. In the Choi state formalism, the distance between the final states obtained by the action of two individual maps, maximized over the set of input states is equivalent to the distance between the Choi states of the concerned maps, i.e.
\begin{eqnarray}
\begin{split}
\max_{\rho}\Big[||\Lambda(\rho)-\Lambda^{M}(\rho)||_1\Big]
\hspace{-2.5cm}\\&
=||(\Lambda\otimes\mathcal{I}_{d_{1}}-\Lambda^{M}\otimes \mathcal{I}_{d_1})\ketbra{\Phi}||_1, \label{choi}
   \end{split}
\end{eqnarray}
 Here, $\ket{\Phi}=\frac{1}{\sqrt{d_1}}\sum_{i}\ket{ii}$ is maximally entangled state between the system and an auxiliary of same dimension that constitutes the Choi state.
 Substituting the form of $\ket{\Phi}$ in equation~\eqref{choi}, we get
\begin{equation*}
\begin{split}
   \max_{\rho} ||\Lambda(\rho)-\Lambda^{M}(\rho)||_1
   \hspace{-2.5cm}\\&=\frac{1}{d_1}\biggl[\left(\Lambda\otimes \mathcal{I}_{d_1}-\Lambda^{M}\otimes \mathcal{I}_{d_1}\right)\sum_{i,j}\ketbra{ii}{jj}\biggr].
    \end{split}
    \end{equation*}
    Similarly for the  perturbed situation, we have,
\begin{eqnarray*}
\hspace{-1.5cm}
\begin{split}
\max_{\rho}\Big[||\Lambda_\epsilon(\rho)-\Lambda^{M}_\epsilon(\rho)||_1\Big]
\hspace{-3.5cm}\\&=||\left(\Lambda_\epsilon\otimes\mathcal{I}_{d_{1}}-\Lambda^{M}_\epsilon\otimes \mathcal{I}_{d_1}\right)\ketbra{\Phi_\epsilon}||_1.
   \end{split}
\end{eqnarray*}
   Where $\ket{\Phi_\epsilon}$ is now the maximally entangled state between the perturbed system and the auxiliary, having the form ${\ket{\Phi_\epsilon}}=\frac{1}{\sqrt{d}}\sum_{i}\ket{i'i}$, with
 $\ket{i'}=\ket{i}+\epsilon\sum_{k\neq i}\frac{\bra{k}H'\ket{i}}{E_i-E_k}\ket{k}$,
 using standard non-degenerate perturbation theory. From this expression we get,
\begin{eqnarray*}
\vspace{2cm}
\begin{split}
\max_{\rho}||\Lambda_\epsilon(\rho)-\Lambda^M_\epsilon(\rho)||_1\\&
\hspace{-2.5cm}=\frac{1}{d_1}\bigg\Vert\biggl[\left(\Lambda\otimes \mathcal{I}_{d_{1}}-\Lambda^M\otimes\mathcal{I}_{d_{1}} \right)\sum_{i,j} \ketbra{i'i}{j'j}\biggr]\bigg\Vert.
    \end{split}
\end{eqnarray*}
Expanding the perturbed basis states $\ket{i'}$ and $\bra{j'}$ in terms of unperturbed basis states, it is easy to check that,
 \begin{eqnarray}
 \vspace{1cm}
 \begin{split}
\max_{\rho}||\Lambda_\epsilon(\rho)-\Lambda^M_\epsilon(\rho)||_1
  \hspace{-3.2cm}\\ &=\frac{1}{d_1}\bigg\Vert\biggl[\biggl(\Lambda\otimes \mathcal{I}_{d_{1}}-\Lambda^M\otimes\mathcal{I}_{d_{1}}\biggr)\\&\biggl(\sum_{i,j} \biggl(\ketbra{ii}{jj} +\epsilon\Big(\sum_{i\neq k}\frac{\bra{k}H'\ket{i}}{E_i-E_k}\ketbra{ki}{jj}\\&+\sum_{l\neq j}\frac{\bra{j}H'\ket{l}}{E_j-E_l}\ketbra{ii}{lj}\Big)\biggr)\biggr)\biggr]\bigg\Vert_1+O(\epsilon^{2}). \label{dist1}
    \end{split}
    \end{eqnarray}
   Assuming the perturbation to be very weak, we neglect the higher order terms and retain only those terms up to first order in $\epsilon$. Employing the triangle inequality $||A+B||\leq ||A||+||B||$ in equation~\eqref{dist1}, we get
    \begin{eqnarray}
 \begin{split}
\max_{\rho}||\Lambda_\epsilon(\rho)-\Lambda^M_\epsilon(\rho)||_1
  \hspace{-3cm}\\ &\leq\frac{1}{d_1}\bigg\Vert\biggl[\biggl(\Lambda\otimes \mathcal{I}_{d_{1}}-\Lambda^M\otimes\mathcal{I}_{d_{1}}\biggr)\sum_{i,j} \ketbra{ii}{jj}\bigg\Vert_1  \\&+\frac{\epsilon}{d_1}\bigg\Vert\biggl(\Lambda\otimes \mathcal{I}_{d_{1}}-\Lambda^M\otimes\mathcal{I}_{d_{1}}\biggr)\\&\biggl(\sum_{i,j,k\neq i}\frac{\bra{k}H'\ket{i}}{E_i-E_k}\ketbra{ki}{jj}\\&+\sum_{i,j,l\neq j}\frac{\bra{j}H'\ket{l}}{E_j-E_l}\ketbra{ii}{lj}\biggr)\biggr]\bigg\Vert_1.
  \label{dist2}
    \end{split}
    \end{eqnarray}
 Using equation~\eqref{dist2}, we can write the expression for $D^{\epsilon}_\Lambda$ as,
 \vspace{-0.3cm}
    \begin{eqnarray}
 \begin{split}
D^\epsilon_\Lambda&= \min_{\Lambda^M_\epsilon}||\Lambda_\epsilon(\rho)-\Lambda^M_\epsilon(\rho)||_1
  \hspace{-2.6cm}\\ &\leq\min_{\Lambda^M_\epsilon}\biggl[\frac{1}{d_1}\bigg\Vert\biggl(\Lambda\otimes \mathcal{I}_{d_{1}}-\Lambda^M\otimes\mathcal{I}_{d_{1}}\biggr)\biggl(\sum_{i,j} \ketbra{ii}{jj}\bigg\Vert_1  \\&+\frac{\epsilon}{d_1}\bigg\Vert\biggl(\Lambda\otimes \mathcal{I}_{d_{1}}-\Lambda^M\otimes\mathcal{I}_{d_{1}}\biggr)\\&\biggl(\sum_{i,j,k\neq i}\frac{\bra{k}H'\ket{i}}{E_i-E_k}\ketbra{ki}{jj}\\&+\sum_{i,j,l\neq j}\frac{\bra{j}H'\ket{l}}{E_j-E_l}\ketbra{ii}{lj}\biggr)\bigg\Vert_1\biggr].
  \label{dist3}
 \end{split}
 \end{eqnarray}
 Since the conditions on parameters of global unitaries derived in Sec.~\ref{Conds} are same for the case of $\mathbb{MTO}$ and $\mathbb{MTO}_\epsilon$, under the assumption that the system has a non-degenerate Bohr spectra (see appendix~\ref{appC} for the proof), minimization over $\Lambda^M_\epsilon$ in the first term of the RHS of the inequality~\eqref{dist3} can be equivalently written as minimization over $\Lambda^M$.
 Also, using the inequality
$\min_x\left(f(x)+g(x)\right)\leq \min_x f(x)+\max_x g(x)$, 
 we can simplify further to obtain, 
    \begin{eqnarray}
 \begin{split}
D^\epsilon_\Lambda&= \min_{\Lambda^M_\epsilon}||\Lambda_\epsilon(\rho)-\Lambda^M_\epsilon(\rho)||_1
 \\ &\hspace{-0.4cm}\leq\min_{\Lambda^M}\biggl[\frac{1}{d_1}\bigg\Vert\biggl(\Lambda\otimes \mathcal{I}_{d_{1}}-\Lambda^M\otimes\mathcal{I}_{d_{1}}\biggr)\biggl(\sum_{i,j} \ketbra{ii}{jj}\bigg\Vert_1\biggr]  \\&\hspace{-0.4cm}+\max_{\Lambda^M}\biggl[\frac{\epsilon}{d_1}\bigg\Vert\biggl(\Lambda\otimes \mathcal{I}_{d_{1}}-\Lambda^M\otimes\mathcal{I}_{d_{1}}\biggr)\\&\biggl(\sum_{i,j,k\neq i}\frac{\bra{k}H'\ket{i}}{E_i-E_k}\ketbra{ki}{jj}\\&+\sum_{i,j,l\neq j}\frac{\bra{j}H'\ket{l}}{E_j-E_l}\ketbra{ii}{lj}\biggr)\bigg\Vert_1\biggr].
 \label{dist4}
 \end{split}
 \end{eqnarray} 
The first term in the inequality~\eqref{dist4} is nothing but $D_\Lambda$, the distance based measure of non-Markovianity corresponding to $\mathbb{TO}$. Thus we have 
\begin{equation}
D^{\epsilon}_\Lambda - D_\Lambda\leq \frac{\epsilon}{d_1} ||\chi||_1.  
\end{equation}
\vspace{-0.5cm}
This completes the proof for the above proposition.
\end{proof}

Unlike the cases of the entanglement and correlation-based measures, where we could obtain examples depicting the enhancement of non-Markovianity under perturbation in the system's Hamiltonian, the following example considering the distance-based measure shows that not always an enhancement in the non-Markovianity can be obtained. The amount of non-Markovianity in the unperturbed and perturbed case in this particular example are almost equal.\\

\noindent \textbf{Example.} Consider the composite system-environment state in $\mathbb{C}^2\otimes\mathbb{C}^2$. i.e. the system as well as the environment is taken to be qubit. The unperturbed system Hamiltonian is given by $H_S=\sigma_z\delta$. The initial state of the system, in the absence of perturbations, is written in the eigenbasis of $H_S$ as $\rho_S=a\ketbra{0}{0}+(1-a)\ketbra{1}{1}$. The environment Hamiltonian is taken as $H_B=\omega\sigma_z\delta$, with $\omega=10$ and the environment is considered to be in the Gibbs state $\tau_B=\frac{e^{-\beta \sigma_z}}{\Tr(e^{-\beta \sigma_z})}$,  with the inverse temperature $\beta=\dfrac{1}{T} \dfrac{k_B}{\hslash\delta}$. The eigenbasis of both $H_S$ and $H_B$ is given by $\{\ket{0},\ket{1}\}$.   We fix the environment temperature to $T=100$ units. 

The global unitary $U_{SB}$ must commute with the total system-environment Hamiltonian $H_T$, if it has to generate a $\mathbb{TO}$. So $U_{SB}$ can be diagonalized in the eigenbasis of $H_T$ as,
\begin{equation*}
\begin{split}
    U_{SB}=e^{-\mathbb{I}\alpha_1}\ket{00}\bra{00}+e^{-\mathbb{I}\alpha_2}\ket{11}\bra{11}\\+e^{-\mathbb{I}\alpha_3}\ket{01}\bra{01}+e^{-\mathbb{I}\alpha_4}\ket{10}\bra{10}.
\end{split}
\end{equation*}

For the unperturbed case, the distance is calculated as
\begin{eqnarray*}
\begin{split}
 D&=\min_{\Lambda^M} ||\left[\left(\Lambda\otimes \mathcal{I}\right)-\left(\Lambda^M\otimes \mathcal{I}\right)\right]\ketbra{\Phi}{\Phi} ||_1,\\   
 &=\frac{1}{2}\min_{\Lambda^M}\biggl[\sum_{i,j=0}^{1}\left(\Lambda-\Lambda^M\right)\ketbra{i}{j}\otimes \ketbra{i}{j}\biggr].
\end{split}
\end{eqnarray*}

Where $\ket{\Phi}=\dfrac{1}{\sqrt{2}}\left(\ket{00}+\ket{11}\right)$ is the maximally entangled state between the system and an auxiliary of same dimension. The action of $\Lambda$ and $\Lambda^M$ are given as $\Lambda\ketbra{i}{j}=\Tr_{B}\left(U_{SB}\ketbra{i}{j}\otimes \tau_BU_{SB}^{\dagger}\right)$ and $\Lambda^M\ketbra{i}{j}=\Tr_{B}\left(U_{SB}^M\ketbra{i}{j}\otimes \tau_BU_{SB}^{M\dagger}\right)$. Here, $U_{SB}^M$ is the global unitary generating the $\mathbb{MTO}$ $\Lambda^M$. The form of $U_{SB}^M$ is given as,
\begin{equation*}
\begin{split}
    U^M_{SB}=e^{-\mathbb{I}\alpha_1'}\ket{00}\bra{00}+e^{-\mathbb{I}\alpha_2'}\ket{11}\bra{11}\\+e^{-\mathbb{I}\alpha_3'}\ket{01}\bra{01}+e^{-\mathbb{I}\alpha_4'}\ket{10}\bra{10}.
\end{split}
\end{equation*}
The parameters of $U_{SB}^M$ are constrained due to the condition of Markovianity.
The constraints are derived by demanding, 
\begin{equation}
U^M_{SB}\rho_S \otimes \tau_B U_{SB}^{M\dagger}=\rho'_S \otimes \tau_B, 
\label{Ma}
\end{equation}
with $\rho'_S=\Tr_B\left(U^M_{SB}\rho_S \otimes \tau_B U_{SB}^{M\dagger}\right)$. 
Putting the form of $U_{SB}^M$ in equation~\eqref{Ma}, the condition on the parameters of $U^M_{SB}$ are obtained as,
\begin{equation}
e^{\mathbb{I}(\alpha_4' -\alpha_1')}=e^{\mathbb{I}(\alpha_2' -\alpha_3')}.  
\label{con}
\end{equation}
Next we calculate $D_\Lambda$ by minimizing over $\mathbb{MTO}$ which entails minimizing over the parameters of the corresponding global unitary $U_{SB}^M$, subject to the condition~\eqref{con}. The parameters $\alpha_k$, with $k=1,2,3,4.$ corresponding to $\mathbb{TO}$ are fixed to $\alpha_1=10^4,\alpha_2=2\times10^4,\alpha_3=3\times10^4,\alpha_4=4\times10^4$. 

For the perturbed case, system's new Hamiltonian is taken to be $H'_S=\left(\sigma_z+\epsilon \sigma_x\right)\delta$.  The initial state of the system is now given as $\rho^{\epsilon}_S=a\ketbra{0'}{0'}+(1-a)\ketbra{1'}{1'}$, where $\{\ket{0'},\ket{1'}\}$ forms the eigenbasis of $H'_S$. The distance-based measure under perturbation is thus,
\begin{equation*}
D_\epsilon=\min_{\Lambda^M} ||\left[\left(\Lambda\otimes \mathcal{I}\right)\\-\left(\Lambda^M\otimes \mathcal{I}\right)\right]\ketbra{{\Phi
_\epsilon}}{\Phi_\epsilon}||_1,
\end{equation*}
Where $\ket{\Phi
_\epsilon}=\dfrac{1}{\sqrt{2}}\left(\ket{0'0}+\ket{1'1}\right)$.  
Here also, we minimize over the parameters of $U_{SB}^M$ subject to the condition,\begin{equation}
U^M_{SB}\rho_S^\epsilon \otimes \tau_B U_{SB}^{M\dagger}=\rho'^\epsilon_S \otimes \tau_B, 
\label{Ma1}
\end{equation}
with $\rho'^\epsilon_S=\Tr_B\left(U^M_{SB}\rho_S \otimes \tau_B U_{SB}^{M\dagger}\right)$.
We again obtain the condition~\eqref{con}.  So, the minimization over the parameters is done such that condition~\eqref{con} is satisfied. We have used numerical non-linear optimization algorithm for computing $D$ and $D_\epsilon$. 
We find that the difference $\Delta D$ is of the order of $10^{-4}$ and less for $\epsilon\in(0.01,0.1)$. Hence, considering the distance-based measure, we get a counter example, which depicts that a perturbation in system's Hamiltonian does not always guarantee an  enhancement in the non-Markovianity.

\section{Conclusion and Discussion}\label{conc}
When the ideal conditions attributed to thermal operation ($\mathbb{TO}$) is not completely satisfied, the resulting operation can be defined to be athermal. An approximate thermal operation ($\mathbb{TO}$$_\epsilon$) is one such athermal operation resulting from the violation of condition of total energy conservation, that is an essential feature of any $\mathbb{TO}$. For a given $\mathbb{TO}$, initial perturbation in the system's Hamiltonian often leads to such violation. 

In this work, we discussed the concept of non-Markovianity in the context of both thermal operations ($\mathbb{TO}$) and the corresponding approximate thermal operations ($\mathbb{TO}$$_\epsilon$). Subsequently we proposed three measures of non-Markovianity applicable for both the types of operations. Our first measure is an entanglement-based measure, that calculates the entanglement generated between the system and the environment under the application of $\mathbb{TO}$ and  $\mathbb{TO}$$_\epsilon$. The second measure is based on the the total correlation generated between the system and the environment by these operations. Lastly, the third measure is distance-based, whereby we compute the distance between the final states obtained by the application of any given $\mathbb{TO}$/  $\mathbb{TO}$$_\epsilon$ to that of the final states obtained upon the action of Markovian thermal operations ($\mathbb{MTO}$)/ approximate Markovian thermal operations ($\mathbb{MTO}_\epsilon$), minimized over all such $\mathbb{MTO}$/$\mathbb{MTO}_\epsilon$, and maximized over all choice of initial states of the system.
 
Our main aim was to analyze how athermality affects non-Markovianity. We related each of the measures of non-Markovianity for a given $\mathbb{TO}$ to its corresponding $\mathbb{TO}$$_\epsilon$. In particular, for the distance-based and entanglement-based measures, we derived upper bounds on the difference of the perturbation-induced non-Markovianity to that of the unperturbed case, considering weak perturbation. We found that the quantities appearing in the bounds implicitly depend on the degree of perturbation $\epsilon$, the characteristic unitary generating the operations and the perturbing Hamiltonian. In case of total correlation-based measure, we were able to exactly compute the difference. 

We presented examples whereby we calculated the difference of the perturbation-induced non-Markovianity to that of the unperturbed case. In our analysis, we took example of states in $\mathbb{C}^2\otimes\mathbb{C}^2$ to calculate the entanglement, and states in $\mathbb{C}^2\otimes\mathbb{C}^3$ for computing the total correlation. For these two cases, we were able to obtain a positive value of the difference, which illustrates that athermality may enhance non-Markovianity. 

The importance of our results lies in the fact that non-Markovianity is a key resource in quantum information science and it is essential to probe its behaviour with disturbances that are unavoidable. Here we demonstrated that even with such disturbances, one may get an amplification of the non-Markovianity. 
\vspace{-0.4cm}
\section{Acknowledgment}
We acknowledge the use of QIClib – a C++ library for quantum information and computation (\url{https://titaschanda.github.io/QIClib}). 
\appendix
\section{Condition on parameters of unitaries for transformation of diagonal elements under $\mathbb{MTO}$}\label{appA}
Let $U_{SB}$ denote the unitary corresponding to a $\mathbb{MTO}$.  
    Transformation of diagonal elements under $\mathbb{MTO}$ is given as
    \begin{equation*}
        U_{SB} \ketbra{i}{i}\otimes \tau_B U_{SB}^{\dagger}=\rho'_S \otimes \tau_B.
    \end{equation*}
Since the unitaries are energy-preserving, we can identify their action on blocks of fixed energies and write
  \begin{equation*}
  \begin{split}
      U_{SB} \ketbra{i}{i}\otimes \tau_B U_{SB}^{\dagger}&=\oplus_{E,E'} U_E \sum_{E_R} P(E_R)\\&\ketbra{i}{i}\otimes \ketbra{E_R}{E_R}  U^{\dagger}_{E'},\\
      &= \sum_{E_R} P(E_R)\oplus_{E,E'} U_E \\&\ketbra{i}{E_R} \otimes \ketbra{i}{E_R}  U^{\dagger}_{E'},
  \end{split}
    \end{equation*}
where the environment state $\tau_B=\sum_{E_R} P(E_R)\ketbra{E_R}{E_R}$ is written in the environment energy basis. 

The action of unitaries on unperturbed system and environment states is given by
\begin{equation*}
    U_E \ket{i}\ket{E_R} =\sum_{j} \alpha_{E_R}^{ji}\ket{j}\ket{E_R+E_i-E_j} \delta_{E,E_R+E_i}.
\end{equation*}
Using this and writing $E_i-E_j=\omega_{ji}$, we get
\begin{equation*}
\begin{split}
    U_{SB}\ketbra{i}{i} \otimes \tau_B U_{SB}^{\dagger}&=\oplus_{E,E'} \sum_{E_R} P(E_R)\sum_{j} |\alpha_{E_R}^{ji}|^{2}\ket{j}\\&\ket{E_R+\omega_{ji}} \otimes \bra{j} \bra{E_R+\omega_{ji}} \\& \delta_{E,E_R+E_i}\delta_{E',E_R+E_i},\\
    &=\sum_{E_R} P(E_R) \sum_{j}|\alpha_{E_R}^{ji}|^{2}\\&\ket{j} \ket{E_R+\omega_{ji}} \otimes \bra{j}\bra{E_R+\omega_{ji}}.
\end{split}
\end{equation*}
Tracing out the environment, we obtain system's new state as
\begin{equation*}
    \rho'_S=\sum_{E_R} P(E_R)\sum_{j}|\alpha_{E_R}^{ji}|^{2} \ketbra{j}{j}.
\end{equation*}
Using the form of $\rho'_S$, the condition for Markovianity can be written as
\begin{equation*}
\begin{split}
   \sum_{E_R} P(E_R) \sum_{j}|\alpha_{E_R}^{ji}|^{2}\ketbra{j}{j}\otimes \ketbra{E_R+\omega_{ji}}{E_R+\omega_{ji}}\\=\sum_{E_R} P(E_R)\sum_{j}|\alpha_{E_R}^{ji}|^{2}\ketbra{j}{j}\otimes \sum_{E_R} P(E_R)\ketbra{E_R}{E_R}. 
\end{split}
\end{equation*}
Writing $P(i \rightarrow j)=\sum_{E_R}P(E_R)|\alpha_{E_R}^{ji}|^{2}$, we have
\begin{equation*}
\begin{split}
 \sum_{E_R} P(E_R) \sum_{j}|\alpha_{E_R}^{ji}|^{2}\ketbra{j}{j}\otimes \ketbra{E_R+\omega_{ji}}{E_R+\omega_{ji}}\\=\sum_{j}P(i \rightarrow j)\ketbra{j}{j}\otimes \sum_{E_R} P(E_R)\ketbra{E_R}{E_R}.   
\end{split}
\end{equation*}
At this point, let us make a variable transformation
\begin{equation*}
  E_R \rightarrow \bar{E}_R=E_R+\omega_{ji}. 
\end{equation*}
So the condition, in terms of this new variable becomes
\begin{equation*}
\begin{split}
  \sum_{\bar{E}_R} P(\bar{E}_R-\omega_{ji}) \sum_{j}|\alpha_{\bar{E}_R-\omega_{ji}}^{ji}|^{2}\ketbra{j}{j}\otimes \ketbra{\bar{E}_R}{\bar{E}_R} \\=\sum_{j}P(i \rightarrow j)\ketbra{j}{j}\otimes \sum_{E_R} P(E_R)\ketbra{E_R}{E_R}.   
\end{split}
\end{equation*}
Since $\ketbra{j}{j} \otimes \ketbra{E_R}{E_R}$s are linearly independent, by equating their coefficients for each $j,E_R$, we get
\begin{equation*}
    |\alpha_{\bar{E}_R-\omega_{ji}}^{ji}|^{2}P(\bar{E}_R-\omega_{ji})=P(E_R)P(i \rightarrow j).
\end{equation*}
Making a variable transformation $\bar{E}_R \rightarrow E_R$, we arrive at the required condition
\begin{equation}
    |\alpha_{E_R}^{ji}|^{2}= \frac{P(E_R+\omega_{ji})P(i \rightarrow j)}{P(E_R)}.
\end{equation}
This completes the proof of equation \eqref{eq:1}.
\section{Condition on parameters of unitaries for transformation of off-diagonal elements under $\mathbb{MTO}$}\label{appB}
Transformation of off-diagonal elements under $\mathbb{MTO}$ is given as
\begin{equation*}
  U_{SB} \ketbra{i}{j}\otimes \tau_B U_{SB}^{\dagger}=\rho'_S \otimes \tau_B. 
\end{equation*}
Proceeding in the same way as in appendix \ref{appA}, we get
\begin{equation*}
\begin{split}
    U_{SB}\ketbra{i}{j}\otimes \tau_B U_{SB}^{\dagger}=\sum_{E_R} P(E_R) \sum_{k,l}\alpha_{E_R}^{ki}(\alpha_{E_R}^{lj})^{*}\\\ket{k} \ket{E_R+\omega_{ki}} \otimes \bra{l} \bra{E_R+\omega_{lj}}.
\end{split}
\end{equation*}
The system's new state is obtained by tracing the environment
\begin{equation*}
\begin{split}
    \rho'_S=\sum_{E_R} P(E_R)\sum_{j}|\alpha_{E_R}^{ji}|^{2}\ketbra{j}{j}\\ \braket{E_R+\omega_{ki}}{E_R+\omega_{lj}}.
\end{split}
\end{equation*}
The inner product $\braket{E_R+\omega_{ki}}{E_R+\omega_{lj}}\neq 0 $ iff $\omega_{ki}=\omega_{lj}$. This leads to  $E_i-E_j=E_k-E_l$. For Hamiltonians with a non-degenerate Bohr spectra, $E_i-E_j\neq E_k-E_l$ for $i\neq k, j\neq l$. So we must have $E_i=E_k, E_j=E_l$. This gives $\omega_{ki}=E_i-E_k=0,\omega_{lj}=E_j-E_l=0$.

The condition of Markovianity can now be written as
\begin{equation*}
\begin{split}
    \sum_{E_R} P(E_R) \alpha_{E_R}^{ii}(\alpha_{E_R}^{jj})^{*}\ketbra{i}{j} \otimes \ketbra{E_R}{E_R}\\
=\sum_{E_R} P(E_R)\alpha_{E_R}^{ii}(\alpha_{E_R}^{jj})^{*}\ketbra{i}{j}\otimes \sum_{E_R} P(E_R)\ketbra{E_R}{E_R}.
\end{split}
\end{equation*}
Writing $\Lambda_{ij}=\sum_{E_R}P(E_R)\alpha_{E_R}^{ii}(\alpha_{E_R}^{jj})^{*}$, we have
\begin{eqnarray*}
\begin{split}
      \sum_{E_R} P(E_R) \alpha_{E_R}^{ii}(\alpha_{E_R}^{jj})^{*}\ketbra{i}{j} \otimes \ketbra{E_R}{E_R}\\&\hspace{-6cm}=\Lambda_{ij}\ketbra{i}{j}\otimes \sum_{E_R} P(E_R)\ketbra{E_R}{E_R}.
\end{split}
\end{eqnarray*}
Equating the coefficients of $\ketbra{i}{j} \otimes \ketbra{E_R}{E_R}$ for each index $i,j,E_R$, we obtain the condition
\begin{equation}
    \alpha_{E_R}^{ii}({\alpha_{E_R}^{jj}}^{*})= \Lambda_{ij}.
\end{equation}
This completes the proof of equation \eqref{eq:2}.
\section{Condition on parameters of unitaries for transformation of elements under $\mathbb{MTO}_\epsilon$}\label{appC}
Here we prove that the constraints on the parameters of unitaries in case of $\mathbb{MTO}_\epsilon$ are the same as the constraints obtained in case of $\mathbb{MTO}$, where we consider only those systems whose Hamiltonians possess a non-degenerate Bohr spectra. Let us consider only the diagonal elements of the system's perturbed density matrix. It can be shown that the same set of conditions are obtained from the off-diagonal elements also. 

The condition of Markovianity implies that
\begin{equation*}
    U_{SB} \ketbra{i'}{i'}\otimes U_{SB}^{\dagger}=\rho'_S\otimes \tau_B.
\end{equation*}
The action of the unitaries on the perturbed states can be obtained by writing the perturbed basis states in terms of unperturbed basis states and using the fact that the unitaries commute with the unperturbed Hamiltonian and hence act on blocks of fixed energy (corresponding to the unperturbed Hamiltonian). Then we have
\begin{equation*}
\begin{split}
   \bigoplus_{E,E'} U_E\biggl(\biggl[\ket{i}+\epsilon \sum_{j\neq i} \frac{\bra{j}H'\ket{i}}{E_i-E_j} \ket{j} \biggr]\\\biggl[\bra{i}+\epsilon \sum_{j\neq i} \frac{\bra{i}H'\ket{j}}{E_i-E_j} \bra{j}\biggr]\\\otimes \sum_{E_R} p\left(E_R\right) \ketbra{E_R}{E_R}\biggr) U^{\dagger}_{E'}=\rho'_S\otimes \tau_B.  
\end{split}
\end{equation*}
We consider first order perturbation theory. Considering feeble perturbation, we keep terms up to first order in $\epsilon$. Using the notations and calculations from appendices \ref{appA} and \ref{appB}, the above expression can be straightforwardly written as 
\begin{equation*}
\begin{split}
 \sum_{E_R} p\left(E_R\right) \sum_{p}|\alpha^{pi}_{E_R}|^{2} \ketbra{p}{p} \otimes \ketbra{E_R+\omega_{pi}}{E_R+\omega_{pi}} \\+\epsilon\biggl[\sum_{j \neq i} \frac{1}{E_i-E_j}\biggl(\bra{i}H'\ket{j} \sum_{p,q}\alpha_{E_R}^{pi}\left(\alpha_{E_R}^{qj}\right)^{*} \ketbra{p}{q}\\\otimes \ketbra{E_R+\omega_{pi}}{E_R+\omega_{qj}}+\bra{j}H'\ket{i} \sum_{p,q}\alpha_{E_R}^{qj}\\\left(\alpha_{E_R}^{pi}\right)^{*} \ketbra{q}{p}\otimes \ketbra{E_R+\omega_{qj}}{E_R+\omega_{pi}}\biggr)\biggr]\\=\biggl[\sum_{j} p(i \rightarrow j)\ketbra{j}{j}+\epsilon \sum_{j \neq i}\frac{1}{E_i-E_j}\\\biggl(\Lambda_{ij}\bra{i}H'\ket{j} \ketbra{i}{j} +\Lambda_{ji}\bra{j}H'\ket{i} \ketbra{j}{i}\biggr)\biggr]\\\otimes \sum_{E_R} p\left(E_R\right) \ketbra{E_R}{E_R},   
\end{split}
\end{equation*}
under the condition $p=i,q=j$ which makes $\omega_{pi}=E_i-E_p=0$, $\omega_{qj}=E_j-E_q=0$. This simplifies the above equation to
\begin{equation*}
    \begin{split}
         \sum_{E_R} p\left(E_R\right) \sum_{j} |\alpha^{ji}_{E_R}|^{2} \ketbra{j}{j} \otimes \ketbra{E_R+\omega_{ji}}{E_R+\omega_{ji}}\\+\epsilon\biggl[\sum_{j \neq i} \frac{1}{E_i-E_j}\biggl(\bra{i}H'\ket{j} \alpha_{E_R}^{ii}\left(\alpha_{E_R}^{jj}\right)^{*}\ketbra{i}{j}\otimes \ketbra{E_R}{E_R}\\+\bra{j}H'\ket{i}\alpha_{E_R}^{jj}\left(\alpha_{E_R}^{ii}\right)^{*}\ketbra{j}{i}\otimes \ketbra{E_R}{E_R}\biggr)\biggr]\\=\sum_{j} p(i \rightarrow j)\ketbra{j}{j}\otimes \ketbra{E_R}{E_R}
    \end{split}
\end{equation*}
\begin{equation*}
    \begin{split}
        +\epsilon \sum_{j \neq i}\frac{1}{E_i-E_j}\biggl(\Lambda_{ij}\bra{i}H'\ket{j} \ketbra{i}{j}\otimes \sum_{E_R} p\left(E_R\right) \ketbra{E_R}{E_R}\\+\Lambda_{ji}\bra{j}H'\ket{i} \ketbra{j}{i}\otimes \sum_{E_R} p\left(E_R\right) \ketbra{E_R}{E_R}\biggr)
    \end{split}
\end{equation*}

Equating the coefficients of $\ketbra{i}{j}\otimes \ketbra{E_R}{E_R}$, we get the first condition,
\begin{equation}
   \alpha_{E_R}^{ii}({\alpha_{E_R}^{jj}}^{*})= \Lambda_{ij}. \label{C1}
\end{equation}
Now let us make a variable change $E_R+\omega_{ji}=\bar{E_R}$. Using this and equating the coefficients of $\ketbra{j}{j}\otimes \ketbra{E_R}{E_R}$,
\begin{equation*}
    p\left(\bar{E_R}-\omega_{ji}\right)\bigg\vert\alpha^{ji}_{\bar{E_R}-\omega_{ji}}\bigg\vert^{2}=p\left(E_R\right)p(i \rightarrow j).
\end{equation*}
By making the transformation $\bar{E}_R \rightarrow E_R$, we get the second condition
\begin{equation}
    |\alpha_{E_R}^{ji}|^{2}= \frac{P(E_R+\omega_{ji})P(i \rightarrow j)}{P(E_R)}.\label{C2}
\end{equation}
Equations~\eqref{C1} and \eqref{C2} are identical to the equations~\eqref{eq:1} and ~\eqref{eq:2}. This proves our assertion.
\vspace{-1cm}
\section{Unitaries generating $\mathbb{TO}$, diagonal in computational basis cannot create entanglement in $\mathbb{C}^2\otimes\mathbb{C}^2$ and $\mathbb{C}^2\otimes\mathbb{C}^3$  }\label{appD}
Suppose the system and the environment are equipped with the respective Hilbert spaces $\mathbb{C}^{d_1}$ and $\mathbb{C}^{d_2}$. Let $\{\ket{i}\}_{i=0}^{d_1-1}$ be the eigenbasis of the system Hamiltonian $H_S$ and $\{\ket{E_R}\}_{E_R=0}^{d_2-1}$ denote the eigenbasis of the environment Hamiltonian $H_B$. Then an eigenbasis of the total Hamiltonian $H_T=H_S \otimes \mathcal{I}_B+\mathcal{I}_S \otimes H_B$ is given by $\{\ket{iE_R}\}$. We consider this basis as the computational basis. An unitary $U_{SB}$ diagonal in this basis would commute with $H_T$ and hence be a generator of $\mathbb{TO}$. However, we show that such an unitary would not be able to create any entanglement between the initially uncorrelated system and the environment in $\mathbb{C}^2\otimes\mathbb{C}^2$ and $\mathbb{C}^2\otimes\mathbb{C}^3$ cases. 

 Let us first consider the $\mathbb{C}^2\otimes\mathbb{C}^2$ case.  Consider an arbitrary $2\otimes2$ matrix in the basis $\{\ket{0},\ket{1}\}$, 
\begin{equation}
\rho_S=\begin{pmatrix}
a & b \\
c & d
\end{pmatrix}.  
\end{equation}
Also, consider an arbitrary $2\times2$ diagonal matrix in the basis $\{\ket{0},\ket{1}\}$,
\begin{equation}
\tau_B=\begin{pmatrix}
e & 0 \\
0 & f
\end{pmatrix}.  
\end{equation}
Let $U_{SB}$ be any arbitrary matrix diagonal in the basis $\{\ket{00},\ket{01},\ket{10},\ket{11}\}$. It can be written as,
\begin{equation}
    U_{SB}=\begin{pmatrix}
    u_1 & 0& 0& 0 \\
0 & u_2 & 0& 0\\
0 & 0 & u_3& 0\\
0 & 0 & 0& u_4  
    \end{pmatrix}.
\end{equation}
After application of $U_{SB}$, the matrix $\rho_S\otimes\tau_B$ becomes,
\begin{equation*}
    \rho_{SB}=U_{SB}\rho_S\otimes\tau_B U_{SB}^{\dagger}.
\end{equation*}
Using the form for $\rho_S$ and $\tau_B$ and simplifying, we can write the expression for $\rho_{SB}$ as
\begin{equation*}
    \rho_{SB}=\begin{pmatrix}
    u_1ae u^*_1 & 0& u_1be u^*_3& 0 \\
0 & u_2afu^*_2 & 0& u_2bfu^*_4\\
u_3ce u^*_1 & 0 & u_3de u^*_3& 0\\
0 & u_4cfu^*_2 & 0& u_4dfu^*_4  
    \end{pmatrix}.
\end{equation*}
The above expression can be cast in a general form as
\begin{equation*}
\rho_{SB}=\begin{pmatrix}
A_1& 0 &A_2&0\\
0& A_3 &0&A_4\\
A_5& 0 &A_6&0\\
0& A_7 &0&A_8
\end{pmatrix}.   
\end{equation*}
Taking partial transpose, we get
\begin{equation*}
\rho_{SB}^{{\Gamma}_A}=\begin{pmatrix}
A_1& 0 &A_5&0\\
0& A_3 &0&A_7\\
A_2& 0 &A_6&0\\
0& A_4 &0&A_8\\
\end{pmatrix}.   
\end{equation*}
The characteristic equation for $\rho_{SB}$ is obtained from the condition
\begin{equation*}
    \begin{vmatrix}
A_1-\lambda & 0 & A_2&0\\
0&A_3-\lambda &0& A_4\\
A_5 &0& A_6-\lambda & 0\\
0&A_7 &0&A_8-\lambda
\end{vmatrix}=0,
\end{equation*}
from which we obtain
\begin{equation*}
\begin{split}
     \{(A_1-\lambda)(A_6-\lambda)-A_2A_5\}\\\{(A_3-\lambda)(A_8-\lambda)-A_4A_7\}=0.
\end{split}
\end{equation*}
Similarly, the characteristic equation for $\rho_{SB}^{\Gamma_A}$ is obtained from the condition
\begin{equation*}
  \begin{vmatrix}
A_1-\lambda & 0 & A_5&0\\
0&A_3-\lambda &0 &A_7\\
A_2 &0& A_6-\lambda & 0\\
0&A_4 &0&A_8-\lambda
\end{vmatrix}  =0,
\end{equation*}
whereby
\begin{equation*}
\begin{split}
     \{(A_1-\lambda)(A_6-\lambda)-A_2A_5\}\\\{(A_3-\lambda)(A_8-\lambda)-A_4A_7\}=0.
\end{split}
\end{equation*}
It can be seen that both $\rho_{SB}$ and $\rho_{SB}^{\Gamma_A}$ have same characteristic equation and hence same set of eigenvalues.
Since this is true for arbitrary matrices $\rho_{SB}$, this is true in particular when $\rho_{SB}$ is a density matrix, describing the composite system-environment state. In that case, all the eigenvalues of $\rho_{SB}$ must be non-negative. Thus, $\rho_{SB}^{\Gamma_A}$ also has all non-negative eigenvalues. Now, the PPT criterion~\cite{PhysRevLett.77.1413,HORODECKI19961} states that any quantum state in $\mathbb{C}^2\otimes\mathbb{C}^2$ and $\mathbb{C}^2\otimes\mathbb{C}^3$ is separable if and only if it has a positive partial transpose. Thus, using the PPT criterion, we conclude that $\rho_{SB}$ is separable. 

We next consider the $\mathbb{C}^2\otimes\mathbb{C}^3$ case. 
 As before, we have,
\begin{equation*}
\rho_S=\begin{pmatrix}
a & b \\
c & d
\end{pmatrix},  
\end{equation*}
in the basis $\{\ket{0},\ket{1}\}$ and
\begin{equation*}
\tau_B=\begin{pmatrix}
e & 0 &0\\
0 & f&0\\
0&0&g
\end{pmatrix},  
\end{equation*}
in the basis $\{\ket{0},\ket{1},\ket{2}\}$.

An arbitrary diagonal matrix $U_{SB}$ can be written in the basis $\{\ket{00},\ket{01},\ket{02},\ket{10},\ket{11},\ket{12}\}$ as, 
\begin{equation*}
    U_{SB}=\begin{pmatrix}
    u_1 & 0& 0& 0&0&0 \\
0 & u_2 & 0& 0&0&0\\
0 & 0 & u_3& 0&0&0\\
0 & 0 & 0& u_4&0&0\\
0&0&0&0&u_5&0\\
0&0&0&0&0&u_6
    \end{pmatrix}.
\end{equation*}

Repeating the same procedure as in $\mathbb{C}^2\otimes\mathbb{C}^2$ case, 
\begin{equation*}
\begin{split}
\rho_{SB}&=U_{SB}\rho_S\otimes\tau_B U_{SB}^{\dagger},\\
   &=\begin{pmatrix}
A_1& 0 &0&A_2&0&0\\
0& A_3 &0&0&A_4&0\\
0&0&A_5& 0 &0&A_6\\
 A_7 &0&0&A_8&0&0\\
 0& A_9 &0&0&A_{10}&0\\
 0&0&A_{11}& 0 &0&A_{12}
\end{pmatrix},  
\end{split}
\end{equation*}
expressed in a general form.

Eigenvalues of $\rho_{SB}^{\Gamma_A}$ and $\rho_{SB}$ can be shown to be equal. So this holds true in the case where $\rho_{SB}$ represents a density matrix. Using the similar arguments as in the previous case, $\rho_{SB}$ is separable by the PPT criterion.

\bibliography{refe}
\end{document}